\documentclass[11pt]{article}
\usepackage[margin=1in]{geometry}
\usepackage{lmodern}
\usepackage[utf8]{inputenc}
\usepackage{latexsym}
\usepackage{amsmath}
\usepackage{amsthm}
\usepackage{amssymb}
\usepackage{graphicx}
\usepackage{ifthen}
\usepackage{calc}
\usepackage{stmaryrd}
\usepackage{enumitem}
\usepackage{hyperref}
\usepackage{alphalph}
\usepackage{stackrel}
\usepackage{tikz-cd}
\newcommand\contract{\mathord{\downarrow}}
\newcommand\expand{\mathord{\uparrow}}

\makeatletter
\renewcommand{\section}{\@startsection
  {section}%
  {1}%
  {0mm}%
  {-1\baselineskip}%
  {0.5\baselineskip}%
  {\normalfont\large\bfseries}%
}
\renewcommand{\subsection}{\@startsection
  {subsection}%
  {2}%
  {0mm}%
  {-1\baselineskip}%
  {0.5\baselineskip}%
  {\normalfont\large\itshape}%
}

\renewcommand{\subsubsection}{\@startsection
  {subsubsection}%
  {3}%
  {0mm}%
  {-1\baselineskip}%
  {0.5\baselineskip}%
  {\normalfont\itshape}%
}

\newsavebox{\tempbox}
\renewcommand{\@makecaption}[2]{
  \vspace{10pt}
  \sbox{\tempbox}{\textbf{#1.} #2}
  \ifthenelse{\lengthtest{\wd\tempbox > \linewidth}}{
    \textbf{#1.} #2\par
  }{
    \begin{center}
      \textbf{#1.} #2
    \end{center}
  }
}

\makeatother

\numberwithin{equation}{section}

\newtheoremstyle{mythm}%
  {}%
  {}%
  {\itshape}%
  {}%
  {\bfseries}%
  {.}%
  {.5em}%
  {\thmname{#1}~\thmnumber{#2}\ifthenelse{\equal{\thmnote{#3}}{}}{}{~(\thmnote{#3})}}%

\newtheoremstyle{mydefn}%
  {}%
  {}%
  {\upshape}%
  {}%
  {\bfseries}%
  {.}%
  {.5em}%
  {\thmname{#1}~\thmnumber{#2}\ifthenelse{\equal{\thmnote{#3}}{}}{}{~(\thmnote{#3})}}%

\newtheoremstyle{myremark}%
  {}%
  {}%
  {\upshape}%
  {}%
  {\itshape}%
  {.}%
  {.5em}%
  {\thmname{#1}~\thmnumber{#2}\ifthenelse{\equal{\thmnote{#3}}{}}{}{~(\thmnote{#3})}}%

\theoremstyle{mythm}
\newtheorem{theorem}{Theorem}[section]
\newtheorem{lemma}[theorem]{Lemma}

\newtheorem{corollary}[theorem]{Corollary}
\newtheorem{fact}[theorem]{Fact}

\theoremstyle{mydefn}

\newtheorem{example}[theorem]{Example}
\theoremstyle{myremark}
\newtheorem{remark}[theorem]{Remark}
\theoremstyle{mythm}

\newcommand{\uend}{\hfill$\lrcorner$}
\newcommand{\uende}{\eqno\lrcorner}

\newcounter{claimcounter}
\newenvironment{claim}[1][]{
  \renewcommand{\proof}{\smallskip\par\noindent\textit{Proof. }}
  \medskip\par\noindent%
  \ifthenelse{\equal{#1}{}}{%
    \setcounter{claimcounter}{0}\refstepcounter{claimcounter}\textit{Claim~\arabic{claimcounter}.}
  }{%
    \ifthenelse{\equal{#1}{resume}}{%
      \refstepcounter{claimcounter}\textit{Claim~\arabic{claimcounter}.}
    }{%
      \textit{Claim~#1.}
    }
  }
}{
  \par\medskip
}

\newcommand{\case}[1]{\par\medskip\noindent\textit{Case #1: }}

\newenvironment{cs}{
  \begin{description}
    \renewcommand{\case}[1]{\item[\itshape\mdseries Case ##1:]}
  }{
  \end{description}
}

\newlist{caselist}{description}{10}
\setlist[caselist]{font=\itshape\mdseries}

\setenumerate[1]{label=(\arabic*)}
\newlist{eroman}{enumerate}{2}
\setlist[eroman,1]{label=(\roman*)}
\setlist[eroman,2]{label=(\alph*)}
\newlist{ealph}{enumerate}{1}
\setlist[ealph]{label=(\Alph*)}

\newcounter{nlistcounter}
\newenvironment{nlist}[1]{
  \renewcommand{\thenlistcounter}{\upshape(#1.\arabic{nlistcounter})}
  \begin{list}{\bfseries\thenlistcounter}{%
      \usecounter{nlistcounter}
      \setlength{\labelwidth}{1.5em}%
      \setlength{\leftmargin}{\labelwidth}%
      \addtolength{\leftmargin}{\labelsep}%
      \setlength{\listparindent}{0em}%
      \setlength{\topsep}{5pt}%
      \setlength{\itemsep}{5pt}%
      \setlength{\parsep}{0pt}%
    }
  }{
  \end{list}
}

\renewcommand{\phi}{\varphi}
\newcommand{\bigmid}{\;\big|\;}

\renewcommand{\mathbf}[1]{\textit{\bfseries #1}}

\renewcommand{\tilde}{\widetilde}
\renewcommand{\hat}{\widehat}
\renewcommand{\bar}{\overline}

\newcommand{\NN}{{\mathbb N}}

\newcommand{\CB}{{\mathcal B}}

\newcommand{\CL}{{\mathcal L}}

\newcommand{\CN}{{\mathcal N}}

\newcommand{\CS}{{\mathcal S}}
\newcommand{\CT}{{\mathcal T}}

\newcommand{\CX}{{\mathcal X}}
\newcommand{\CY}{{\mathcal Y}}
\newcommand{\CZ}{{\mathcal Z}}

\DeclareMathOperator{\sep}{\textsc{Sep}}
\DeclareMathOperator{\order}{\textsc{Order}}
\DeclareMathOperator{\tangorder}{\textsc{TangOrder}}
\DeclareMathOperator{\size}{\textsc{Size}}

\DeclareMathOperator{\trunc}{\textsc{Truncation}}
\DeclareMathOperator{\find}{\textsc{Find}}

\newcommand{\dagle}{\trianglelefteq}
\newcommand{\dagsle}{\lhd}
\newcommand{\dagri}{\trianglerighteq}

\usepackage{color}
\definecolor{gruen}{rgb}{0,0.6,0.2}

\newcounter{rbcounter}
\newcommand{\KT}{{\mathfrak T}}
\newcommand{\ord}{\operatorname{ord}}

\renewcommand{\vec}[1]{\overrightarrow{#1}}

\title{Computing with Tangles\thanks{A preliminary version of this
    paper appeared in the Proceedings of the 47th Annual Symposium on
    the Theory of Computing (STOC 2015).}}
\author{Martin Grohe and Pascal Schweitzer\\RWTH Aachen University\\\normalsize\{grohe,schweitzer\}@informatik.rwth-aachen.de}
\date{}

\begin{document}

\maketitle

\begin{abstract}
  Tangles of graphs have been introduced by Robertson and Seymour in
  the context of their graph minor theory. Tangles may be viewed as
  describing ``$k$-connected components'' of a graph (though in a
  twisted way). An interesting aspect of tangles is that they cannot only be
  defined for graphs, but more generally for arbitrary connectivity
  functions (that is, integer-valued submodular and symmetric set functions).

  However, tangles are difficult to deal with algorithmically. To
  start with, it is unclear how to represent them, because they are
  families of separations and as such may be exponentially large. Our
  first contribution is a data structure for representing and accessing
  all tangles of a graph up to some fixed order.

  Using this data structure, we can prove an algorithmic version of a
  very general structure theorem due to Carmesin, Diestel, Hamann and
  Hundertmark (for graphs) and Hundertmark (for arbitrary connectivity
  functions) that yields a canonical tree decomposition whose parts
  correspond to the maximal tangles. This may be viewed as a
  generalisation of the decomposition of a graph into its 3-connected
  components. 
\end{abstract}

\section{Introduction}
Tangles are strange objects---yet they are very
useful. Tangles of graphs have been introduced by Robertson and
Seymour \cite{gm10}, and they play an important role in their graph
minor theory (see, e.g. \cite{gm16}). Intuitively, tangles of order
$k$ may be viewed as descriptions of the ``$k$-connected components'' of a
graph. Recall that every graph has a nice and simple decomposition
into its 2-connected components, which are induced subgraphs of the
graph. A graph also has a well-defined decomposition into its
3-connected components. However, the 3-connected components
are not necessarily subgraphs of the graph, but only topological
subgraphs; they may contain so-called ``virtual edges'' not present in
the graph. It is not surprising that for 4-connected components the
situation becomes even more complicated. In fact, for $k\ge 4$ there
is no clear-cut notion of a $k$-connected component of a graph. Tangles
may be seen as one attempt towards such a notion. (Another one, which is
related to tangles, has recently been proposed by
Carmesin, Diestel, Hundertmark and Stein~\cite{cardiehunstei14}.)

The idea of tangles is to describe a ``region'' in a graph (maybe the
presumed $k$-connected component) by pointing to it. Formally,
this means that a tangle assigns to each low-order separation of the
graph (separation of order less than $k$ in the case of $k$-connected
components) a ``big side'', which is the side where the ``region''
described by the tangle is to be found. To turn this into a meaningful
definition, the assignment of ``big sides'' to the separations has to
satisfy certain consistency conditions. Note, however, that a
``region'' described in this way is elusive and does not necessarily
correspond to a subgraph or even just a subset of the vertex set or
edge set of the graph, because the intersection of the ``big sides''
of all separations may be empty. This way of describing a region
may be viewed as dual to a more direct description as (something resembling)
a subgraph. Indeed, tangles are dual to a form of graph decompositions
known as branch decompositions, which are closely related to the
better-known tree decompositions, in a precise technical sense
(see the Duality Lemma~\ref{lem:duality} due to \cite{gm10}).

Carmesin, Diestel, Hamann, and Hundertmark~\cite{cardiehar+13a},
extending
earlier work of Robertson and Seymour~\cite{gm10} (also see Reed~\cite{ree97}), proved that every
graph has a canonical decomposition into parts corresponding to its
tangles of order at most $k$, that is, a canonical decomposition into
its $k$-connected components. Here ``canonical'' means that an
isomorphism between two graphs can be extended to an isomorphism
between their decompositions.

What makes tangles even more interesting, and was actually our
motivation to start this work, is that they can be defined in a very
abstract setting that applies to various ``connectivity measures'' not
only for graphs but also for other structures such as hypergraphs and
matroids. Two examples are the
``cut-rank'' measure for graphs, which leads to the notion of ``rank
width'' \cite{oum05,oumsey06}, and the connectivity function of a
matroid (see, for example, \cite{oxl11}). Such ``connectivity measures'' can be
specified by a symmetric and submodular function defined on  the
subsets of some ground set. Tangles give us an abstract notion of
``k-connected components'' with respect to these connectivity
measures. Hundertmark~\cite{hun11} generalised the decomposition theorem
from graphs to this abstract setting, giving us canonical
decompositions of structures into parts corresponding to the tangles
with respect to arbitrary connectivity measures. (Earlier, Geelen, Gerards, and
Whittle~\cite{geegerwhi09} had already shown the existence of such
decompositions, but not canonical ones.)

However, these decomposition theorems are pure existence theorems;
they are not algorithmic.\footnote{The ``algorithms'' in the title of
  \cite{cardiehar+13a} refer to general strategies for choosing
  the set of separations in a decomposition, they are not
  implementable algorithms.} In fact, it is not clear how to
efficiently compute with tangles at all. Tangles are defined as
families of separations of a ground set (such as the edge set of a graph)
and as such may be exponentially large; this means that a priori there
is only a doubly-exponential upper bound on the number of
tangles. Remarkably, the decomposition theorems mentioned above imply
that for each $k$ there is only a linear number of tangles of order
$k$. But this still does not tell us how to identify them. Our first main
contribution is a data structure that represents all tangles of some
order $k$ and provides us with a ``membership oracle'' for these tangles, that is,
for a given separation we can ask which side of the separation is the
``big side'' with respect to a given tangle. The data structure as
well as the membership oracle can be implemented in polynomial time
for every fixed order $k$. 

Using this data structure, we can then
prove that a canonical decomposition of a structure into parts
corresponding to the tangles of order at most $k$ can be computed in
polynomial time (again for fixed $k$). This is our second main
result. Proving these results, we devise a number of 
algorithmic subroutines that may be useful in other contexts. For
example, we show how to find canonical ``leftmost minimum separations''
between a set and a tangle. All our results apply in the most general
setting where we are only given oracle access to an arbitrary integer-valued
symmetric and submodular connectivity function. Our algorithms rely on
the minimisation of submodular functions
\cite{iwaflefuj01,schrijver00}. We build on algorithmic ideas for computing branch decompositions of connectivity
functions due to Oum and
Seymour~\cite{oumsey07,oumsey06,oumsey06a}. Furthermore, the duality
between tangles and branch decompositions plays an important role in
our proofs.

Our main motivation for this work is isomorphism testing and
canonisation. It is almost self-evident that a canonical decomposition
of a structure into highly connected parts may be useful in this
context, in a similar way as the decomposition of a graph into its
3-connected components is essential for planar graph isomorphism
testing. Tangles and the decompositions of graphs they induce have
already played an important role in the recent polynomial-time
isomorphism test for graph classes with excluded topological subgraphs by
Marx and the first author of this paper~\cite{gromar12}. However,
there it was sufficient to only work with a specific type of tangles
that can be represented by so called ``well-linked'' or
``unbreakable'' sets of bounded size; this way the computational
problems addressed in the present paper could be
circumvented. In other settings, this is impossible. In
\cite{grosch15b}, we use a ``directed version''
of our canonical
decomposition (Theorem~\ref{theo:dcandec}) to design a
polynomial isomorphism test for graph classes of bounded rank width.

A different, more speculative, application of the techniques developed
here may be the logical definability of decompositions, which is
related to well-known open problems such as Seese's conjecture
\cite{see91,cououm07} and the definability of tree decompositions in
monadic second-order logic (see \cite{coueng12}). Definable
decompositions are necessarily canonical, because logics define
isomorphism-invariant objects. Even though our algorithms may
not translate into logical definitions directly, our constructive
arguments may help to come up with such definitions.

While these applications rely on the \emph{canonicity} of the
decompositions we obtain, our data structure for tangles (which is not
canonical) may have other algorithmic applications. Tangles have
already played a role in algorithmic structural graph theory (for
example in \cite{demhajkaw05,kawwol11,gromar12,grokawree13}), but due
to the non-constructive nature of tangles, this role is usually
indirect or implicit. Having direct access to all tangles may
facilitate new applications.

\section{Connectivity Functions, Tangles, and Branch Decompositions}
\label{sec:tangles}
In this section we introduce connectivity functions, branch
decompositions, and tangles and prove some basic results about
them. All lemmas in this section, except the Exactness Lemma and the
Duality Lemma from \cite{gm10}, are simple and can be proved by
standard techniques.

A \emph{connectivity function} on a finite set $U$ is a symmetric and
submodular function
$\kappa\colon 2^U\to\NN$ with $\kappa(\emptyset)=0$. \emph{Symmetric} means that $\kappa(X)=\kappa(\bar
X)$ for all $X\subseteq U$; here and whenever the ground set $U$ is clear from the context we
write $\bar X$ to denote $U\setminus
X$, the complement of~$X$. \emph{Submodular} means that $\kappa(X)+\kappa(Y)\ge\kappa(X\cap
Y)+\kappa(X\cup Y)$ for all $X,Y\subseteq U$. Observe that a symmetric
and submodular set function is also \emph{posimodular}, that is, it
satisfies $\kappa(X)+\kappa(Y)\ge\kappa(X\setminus
Y)+\kappa(Y\setminus X)$ (apply submodularity to $X$ and $\bar
Y$). 

Note that it is no real restriction to assume that the range of
$\kappa$ is the set $\NN$ of nonnegative integers instead of the set
$\mathbb Z$ of all integers and that $\kappa(\emptyset)=0$. For an arbitrary symmetric  and
submodular function
$\kappa'\colon 2^U\to\mathbb Z$, we can work with the function
$\kappa$ defined by
$\kappa(X):=\kappa'(X)-\kappa'(\emptyset)$ instead. It clearly satisfies
$\kappa(\emptyset)=0$. To see that it is nonnegative, observe that for
all $X\subseteq U$,
\[
\kappa'(X)=\frac{1}{2}\big(\kappa'(X)+\kappa'(\bar X)\big)\ge
\frac{1}{2}\big(\kappa'(\emptyset)+\kappa'(U)\big)=\kappa'(\emptyset).
\]
Here the first an third inequalities hold by symmetry, and the second
inequality holds by submodularity.

\begin{example}
  Maybe the most obvious example of a connectivity function is the
  ``cut function'' in a graph $G$. For all subsets $X,Y\subseteq V(G)$, we let
  $E_{X,Y}$ be the set of all edges of $G$ with one endvertex in $X$
  and one endvertex in $Y$. We define a connectivity
  function $\gamma_G$ on $V(G)$ by $\gamma_G(X)=|E_{X,\bar X}|$.
  \uend
\end{example}

\begin{example}[\mbox{Robertson and Seymour~\cite{gm10}}]\label{exa:1}
  Let $G$ be a graph. The \emph{boundary}~$\partial X$ of an edge set $X\subseteq
  E(G)$ is the set of all vertices of $G$ incident with an edge in $X$
  and an edge in $\bar X=E(G)\setminus X$. We define a connectivity
  function $\kappa_G$ on $E(G)$ by $\kappa_G(X)=|\partial X|$. This is
  the connectivity function Robertson and Seymour~\cite{gm10} originally developed
  their theory of tangles for.
  \uend
\end{example}

\begin{example}[Oum and Seymour~\cite{oumsey06}]\label{exa:2}
    Let $G$ be a graph. For all subsets $X,Y\subseteq V(G)$, we let
  $M=M_G(X,Y)$ be the $X\times Y$-matrix over the 2-element field
  $\mathbb F_2$ with entries $M_{xy}=1\iff xy\in E(G)$.  Now we
  define a connectivity function $\rho_G$ on $V(G)$ by letting
  $
  \rho_G(X)
  $, known as the \emph{cut rank} of $X$, 
   be the row rank of the matrix $M_G(X,\bar X)$. This connectivity
   function was introduced by Oum and Seymour to define the \emph{rank
     width} of graphs, which approximates the \emph{clique width}, but
   has better algorithmic properties.
   \uend
\end{example}

\begin{example}\label{exa:3}
  Let $M$ be a matroid with ground set $E$ and rank function $r$. (The rank of a set
  $X\subseteq E$ is defined to be the maximum size of an independent set
  contained in $X$.) The connectivity function of $M$ is the set
  function $\kappa_M:2^E\to\NN$ defined by $\kappa_M(X)=r(X)+r(\bar
  X)-r(E)$ (see, for example, \cite{oxl11}).
  \uend
\end{example}

For the rest of this section, let $\kappa$ be a connectivity function
on a finite set $U$.
We often think of a subset $Z\subseteq U$ as a \emph{separation} of $U$
into $Z$ and $\bar Z$ and of $\kappa(Z)$ as the \emph{order} of this
separation; consequently, we also refer to $\kappa(Z)$ as the
\emph{order of $Z$}. For disjoint sets $X,Y\subseteq U$, an
\emph{$(X,Y)$-separation} is a set $Z\subseteq U$ such that
$X\subseteq Z\subseteq \bar Y$. Such a separation $Z$ is minimum if
its order is minimum, that is, if $\kappa(Z)\le\kappa(Z')$ for all
$(X,Y)$-separations $Z'$. It is an easy consequence of the
submodularity of $\kappa$ that there is a unique minimum
$(X,Y)$-separation $Z$ such that $Z\subseteq Z'$ for all other minimum
$(X,Y)$-separations $Z'$. We call $Z$ the \emph{leftmost minimum
  $(X,Y)$-separation}. There is also a unique \emph{rightmost minimum
  $(X,Y)$-separation}, which is easily seen to be the complement of
the leftmost minimum $(Y,X)$-separation.

\subsection{Tangles}

A \emph{$\kappa$-tangle} of order $k\ge0$ is a set $\CT\subseteq 2^U$
satisfying the following conditions.\footnote{Our definition of tangle
  is ``dual'' to the one mostly found in the literature, e.g.~\cite{geegerwhi09,hun11}. In our
  definition, the ``big side'' of a separation belongs to the tangle,
  which seems natural if one thinks of a tangle as ``pointing to a
  region'' (as described in the introduction). But of course the
  definitions yield equivalent theories.}
  \begin{nlist}{T}
  \setcounter{nlistcounter}{-1}
  \item\label{li:t0}
    $\kappa(X)<k$ for all $X\in\CT$, 
  \item\label{li:t1}
    For all $X\subseteq U$ with $\kappa(X)<k$, either $X\in\CT$ or
    $\bar X\in\CT$.
  \item\label{li:t2}
    $X_1\cap X_2\cap X_3\neq\emptyset$ for all $X_1,X_2,X_3\in\CT$.
  \item\label{li:t3}
    $\CT$ does not contain any singletons, that is, $\{a\}\not\in\CT$ for all $a\in U$.
\end{nlist}
We denote the order of a $\kappa$-tangle $\CT$ by $\ord(\CT)$.

\begin{example}
  Let $G$ be a graph and $H$ a $3$-connected subgraph of $G$. For a set
  $X\subseteq E(G)$ we let $V(X)$ be the set of all the endvertices of
  edges in $X$. So $\partial(X)=V(X)\cap V(\bar X)$. Let $\CT$ be the
  set of all subsets $X\subseteq E(G)$ such that
  $\kappa_G(X)=|\partial(X)|<3$ (see Example~\ref{exa:1})  and $V(H)\subseteq V(X)$. It is not
  difficult to prove that $\CT$ is a $\kappa_G$-tangle of order $3$ in $G$.

  Actually, essentially the same argument works if $H$ is a
  subdivision of a 3-connected graph and $\CT$ the set of all
  $X\subseteq E(G)$ such that $\kappa_G(X)=|\partial(X)|<3$ and $X$
  contains all vertices of $H$ of degree at least $3$.
  \uend
\end{example}

\begin{example}[\mbox{Robertson and Seymour~\cite[Section 7]{gm10}}]\label{exa:4}
  Let $G$ be a graph and $H\subseteq G$ a $(k\times k)$-grid. Let
  $\CT$ be the set of all $X\subseteq E(G)$ such that $\kappa_G(X)<k$
  and $X$ contains all edges of some row of
  the grid. Then $\CT$ is a $\kappa_G$-tangle of order $k$.
  \uend
\end{example}

Let
$\CT,\CT'$ be $\kappa$-tangles. If $\CT'\subseteq\CT$, we say that
$\CT$ is an \emph{extension} of $\CT'$. The tangles $\CT$ and $\CT'$
are \emph{incomparable} (we write $\CT\bot\CT'$) if neither is an
extension of the other. 
The \emph{truncation} 
of $\CT$ to order $k\le\ord(\CT)$ is the set
$
\{X\in\CT\mid\kappa(X)<k\},
$
which is obviously a tangle of order $k$. Observe that if $\CT$ is
an extension of $\CT'$, then $\ord(\CT')\le\ord(\CT)$, and $\CT'$ is
the truncation of $\CT$ to order $\ord(\CT')$. 

\begin{remark}\label{rem:tangle-order1}
There is a small technical issue that one needs to be aware of, but
that never causes any real problems: if we view tangles as families of
sets, then their order is not always well-defined. Indeed, if there is
no set $X$ of order $\kappa(X)=k-1$, then a tangle of order $k$
contains exactly the same sets as its truncation to order $k-1$. In such a situation, we
have to explicitly annotate a tangle with its order, formally viewing a
tangle as a pair $(\CT,k)$ where $\CT\subseteq 2^U$ and $k\ge 0$. We always view a
tangle of order $k$ and its truncation to order $k-1$ as distinct
tangles, even if they contain exactly the same sets.%
\end{remark}

A \emph{$(\CT,\CT')$-separation} is a set $Z\subseteq U$ such that
$Z\in\CT$ and $\bar Z\in\CT'$. Obviously, if $Z$ is a
$(\CT,\CT')$-separation then $\bar Z$ is a
$(\CT',\CT)$-separation. Observe that there is a
$(\CT,\CT')$-separation if and only if $\CT$ and $\CT'$ are
incomparable. The \emph{order} of a $(\CT,\CT')$-separation $Z$ is
$\kappa(Z)$. A $(\CT,\CT')$-separation $Z$ is \emph{minimum} if its
order is minimum.

\begin{lemma}\label{lem:lm-tansep}
  Let $\kappa$ be a connectivity function on a set $U$, and let
  $\CT,\CT'$ be incomparable tangles. Then there is a (unique) minimum
  $(\CT,\CT')$-separation $Z(\CT,\CT')$ such that
  $Z(\CT,\CT')\subseteq Z'$ for all minimum $(\CT,\CT')$-separations
  $Z'$.

  We call $Z(\CT,\CT')$ the \emph{leftmost minimum $(\CT,\CT')$-separation}.
\end{lemma}

\begin{proof}%
    Let $Z=Z(\CT,\CT')$ be a minimum $(\CT,\CT')$-separation of minimum size $|Z|$,
  and let $Z'$ be another minimum $(\CT,\CT')$-separation. We shall
  prove that $Z\subseteq Z'$. 

  Let $k:=\kappa(Z)=\kappa(Z')<\min\{\ord(\CT),\ord(\CT')\}$. We claim
  that
  \begin{equation}
    \label{eq:cantan1}
    \kappa(Z\cup Z')\ge k.
  \end{equation}
  Suppose for contradiction that $\kappa(Z\cup Z')<k$. Then $Z\cup
  Z'\in\CT$, because $\kappa(Z\cup Z')<\ord(\CT)$ and $\bar{Z\cup
    Z'}\cap Z=\emptyset$. Furthermore, $\bar{Z\cup Z'}\in\CT'$,
  because $\kappa(\bar{Z\cup Z'})=\kappa(Z\cup Z')<\ord(\CT')$ and
  $\bar Z\cap\bar Z'\cap (Z\cup Z')=\emptyset$. Thus $Z\cup Z'$ is a
  $(\CT,\CT')$-separation of smaller order than $Z$. This
  contradicts the minimality of $Z$.

  By submodularity, 
  \begin{equation}
    \label{eq:cantan2}
    \kappa(Z\cap Z')\le k.
  \end{equation}
  A similar argument as above shows that $Z\cap Z'$ is a
  $(\CT,\CT')$-separation. Thus, by the minimality of $Z$, we have
  $\kappa(Z\cap Z')=k$, and by the minimality of $|Z|$ we have
  $|Z|\le|Z\cap Z'|$. This implies $Z=Z\cap Z'$ and thus $Z\subseteq Z'$.
\end{proof}

Now that we have defined $(X,Y)$-separations for sets $X,Y$ and
$(\CT,\CT')$-separations for tangles $\CT,\CT'$, we also need to
define combinations of both. For a $\kappa$-tangle $\CT$ and a set
$X\subseteq U$ a \emph{$(\CT,X)$-separation} is a set $Z\in\CT$ such
that $Z\subseteq\bar X$. Note that even if $X\not\in\CT$, such a
separation does not always exist.  A $(\CT,X)$-separation $Z$ is
\emph{minimum} if its order is minimum. $Z$ is \emph{leftmost minimum}
(\emph{rightmost minimum}) if $Z\subseteq Z'$ for all minimum
$(\CT,X)$-separations $Z'$.  By essentially the same submodularity
arguments as those in the proof of Lemma~\ref{lem:lm-tansep}, it is
easy to prove that if there is a $(\CT,X)$-separation then there is a
unique leftmost minimum $(\CT,X)$-separation and a unique rightmost
minimum $(\CT,X)$-separation.

\subsection{Branch Decompositions and Branch Width}
A \emph{cubic tree} is a
tree where every node that is not a leaf has degree~$3$. An
\emph{oriented edge} of a tree $T$ is a pair $(s,t)$, where $st\in
E(T)$. We denote the set of all oriented edges of $T$ by $\vec E(T)$
and the set of leaves of $T$ by $L(T)$.  A \emph{partial
  decomposition} of $\kappa$ is a pair $(T,\tilde\xi)$, where $T$ is a cubic
tree and $\tilde\xi\colon\vec E(T)\to 2^U$ such that $\tilde\xi(s,t)=\bar{\tilde\xi(t,s)}$
for all $st\in E(T)$ and $\tilde\xi(s,t_1)\cup\tilde\xi(s,t_2)\cup\tilde\xi(s,t_3)=U$
for all non-leaf nodes $s\in V(T)$ with neighbours $t_1,t_2,t_3$. The
partial decomposition is \emph{exact} if for all non-leaf nodes $s\in
V(T)$ with neighbours $t_1,t_2,t_3$ the sets $\tilde\xi(s,t_1),
\tilde\xi(s,t_2),\tilde\xi(s,t_3)$ are mutually disjoint. To simplify the
notation, for all leaves $u\in L(T)$ we let $\xi(u):=\tilde\xi(t,u)$, where
$t$ is the neighbour of $u$ in~$T$. Note that for an exact partial
decomposition, the values $\xi(u)$ at the leaves $u$ determine all other
values $\tilde\xi(s,t)$, and the sets $\xi(u)$ for the leaves $u$ form
a partition of $U$. We say that a partial
decomposition $(T,\tilde\xi)$ is \emph{over} a set $\CX\subseteq 2^U$ if
$\xi(u)\in\CX$ for all leaves $u\in L(T)$. Finally, a \emph{branch
  decomposition} of $\kappa$ is an exact partial decomposition over the set
$\CS_U=\{\{u\}\mid u\in U\}$ of all singletons.\footnote{Since a branch
  decomposition is exact, it suffices to specify the values $\xi(t)$
  at the leaves, and as these values are singletons, we can also
  define $\xi$ as a mapping from $L(T)$ to $U$, which must be
  bijective.}

The \emph{width} of a partial decomposition $(T,\tilde\xi)$ is the maximum of the
values $\kappa(\tilde\xi(s,t))$ for $(s,t)\in \vec E(T)$. The \emph{branch width} of
$\kappa$ is the minimum of the widths of all branch decompositions of
$\kappa$.

\begin{example}
  The branch width of $\kappa_G$ (Example~\ref{exa:1}) is known as the
  \emph{branch width} of the graph $G$. The branch width of $\rho_G$
  (Example~\ref{exa:2}) is known as the \emph{rank width} of 
  $G$. The branch width of $\kappa_M$ (Example~\ref{exa:3}) is known
  as the branch width of the matroid $M$.
  \uend
\end{example}

The following Exactness Lemma and Duality Lemma due to Robertson and
Seymour~\cite{gm10} are fundamental results relating decompositions
and tangles. For the reader's convenience, we include proofs of these
lemmas.

\begin{lemma}[\mbox{Exactness Lemma \cite{gm10}}]
  Let $(T,\tilde\xi)$ be a partial decomposition of $\kappa$. Then
  there is a function $\tilde\xi':\vec E(T)\to 2^U$ such that
  $(T,\tilde\xi')$ is an exact partial decomposition of $\kappa$
  satisfying $\kappa(\tilde\xi'(s,t))\le \kappa(\tilde\xi(s,t))$ for
  all $(s,t)\in\vec E(T)$ and $\xi'(t)\subseteq\xi(t)$ for all leaves
  $t\in L(T)$.
\end{lemma}

\begin{proof}
    It is convenient in this proof to work with binary rooted trees
  instead of cubic trees. To turn $T$ into a binary tree $T^b$, we choose
  an arbitrary edge, say, $s^bt^b\in E(T)$, and subdivide it,
  inserting a new node $r^b$. We make $r^b$ the root of our new binary
  tree $T^b$. Now we can ``push'' the mapping $\tilde\xi$ from the
  edges of the tree to the nodes. That is, we define a mapping
  $\xi^b:V(T)\to 2^U$ by 
  \begin{itemize}
  \item $\xi^b(r^b):=U$;
  \item $\xi^b(s^b):=\tilde\xi(t^b,s^b)$ and
    $\xi^b(t^b):=\tilde\xi(s^b,t^b)$;
  \item $\xi^b(t):=\tilde\xi(s,t)$ for all nodes $t\in
    V(T^b)\setminus\{r^b,s^b,t^b\}$ with parent $s$.
  \end{itemize}
  Observe that for all nodes $s\in V(T^b)$ with children $t_1,t_2$ it
  holds that
  $
  \xi^b(s)\subseteq\xi^b(t_1)\cup\xi^b(t_2).
  $
  Furthermore, the trees $T$ and $T^b$ have the same leaves and
  for all leaves $t$ we have $\xi^b(t)=\xi(t)$.

  In general, if $(T',r')$ is rooted binary tree and $\xi':V(T')\to 2^U$
  such that $\xi'(r')=U$ and $\xi'$ satisfies 
    \begin{equation}
    \label{eq:3}
    \xi'(s)\subseteq\xi'(t_1)\cup\xi'(t_2),
  \end{equation}
  we call
  $(T',r',\xi')$ a \emph{binary partial decomposition}. The
  decomposition is \emph{exact} at a node $s$ with children $t_1,t_2$
  if \eqref{eq:3} can be strengthened to
   \begin{equation}
    \label{eq:4}
    \xi'(s)=\xi'(t_1)\uplus\xi'(t_2),
  \end{equation}
  where $\uplus$ denotes disjoint union.
  The decomposition $(T',r',\xi')$ is \emph{exact} if it is exact at
  all inner nodes.
  We have seen how we can turn a partial decomposition into a binary
  partial decomposition. It is also easy to reverse the construction
  and turn an exact binary
  partial decomposition $(T',r',\xi')$ into an exact partial
  decomposition $(T'',\tilde\xi'')$.

  We will
  iteratively construct a sequence $\xi_1,\ldots,\xi_m$ of mappings
  from $V(T^b)$ to $2^U$ such that
  $(T^b,r^b,\xi_1),\ldots,(T^b,r^b,\xi_m)$ are binary partial
  decompositions satisfying the following invariants for all $i\in[m-1]$ and nodes $t\in V(T^b)$:
  \begin{eroman}
  \item $\kappa(\xi_{i+1}(t))\le \kappa(\xi_{i}(t))$;
  \item either $\xi_{i+1}(t)\subseteq\xi_{i}(t)$ or
    $\kappa(\xi_{i+1}(t))< \kappa(\xi_{i}(t))$;
  \item if $t$ is a leaf of $T^b$ then
    $\xi_{i+1}(t)\subseteq\xi_{i}(t)$.
  \end{eroman}
  Furthermore, the decomposition $(T^b,r^b,\xi_m)$ will be
  exact.
  
  We let $\xi_1:=\xi^b$.
  In the inductive step, we assume that we have
  defined $\xi_i$. If $(T^b,r^b,\xi_i)$ is exact, we let
  $m:=i$ and stop the construction. Otherwise, we pick an arbitrary node $s\in V(T^b)$ with
  children $t_1,t_2$ such that 
  $(T^b,r^b,\xi_i)$ is not exact at $s$, that is, either
  $\xi_i(s)\subset\xi_i(t_1)\cup\xi_i(t_2)$ or
  $\xi_i(t_1)\cap\xi_i(t_2)\neq\emptyset$. We let $X:=\xi_i(s)$ and
  $Y_p:=\xi_i(t_p)$ for $p=1,2$.

  In each of
  the following cases, we only modify $\xi_i$ at the nodes
  $s,t_1,t_2$ and let $\xi_{i+1}(u):=\xi_i(u)$ for all
  $u\in V(T^b)\setminus\{s,t_1,t_2\}$.

  \smallskip\noindent
  \textit{Case 1:}
  $X\subset Y_1\cup Y_2$.
  
  \smallskip\noindent
  \textit{Case 1a:}
  $\kappa(X\cap Y_p)\le\kappa(Y_p)$ for $p=1,2$.\\
  We let
  $\xi_{i+1}(s):=\xi_i(s)$ and
  $\xi_{i+1}(t_p):=X\cap Y_p$ for $p=1,2$. 

  Note that in this case we have
  $\kappa(\xi_{i+1}(u))\le \kappa(\xi_{i}(u))$ and $\xi_{i+1}(u)\subseteq\xi_{i}(u)$ for all
  nodes $u$ and either $\xi_{i+1}(t_1)\subset\xi_{i}(t_1)$
  or $\xi_{i+1}(t_2)\subset\xi_{i}(t_2)$.

  \smallskip\noindent
  \textit{Case 1b:}
  $\kappa(X\cap Y_p)>\kappa(Y_p)$ for some $p\in\{1,2\}$.\\
  For
  $p=1,2$, we let $\xi_{i+1}(t_p):=Y_p$.

  By submodularity we
  have $\kappa(X\cup Y_p)<\kappa(X)$ for some $p\in\{1,2\}$. 
  If $\kappa(X\cup Y_1)<\kappa(X)$ we let $\xi_{i+1}(s):=X\cup
  Y_1$, and otherwise we let $\xi_{i+1}(s):=X\cup Y_2$.  

  Note that in this case we have
  $\kappa(\xi_{i+1}(u))\le \kappa(\xi_{i}(u))$ for all
  nodes $u$ and $\kappa( \xi_{i+1}(s))<\kappa(
  \xi_{i}(s))$ and $\xi_{i+1}(u)=\xi_{i}(u)$ for all
  nodes $u\neq s$.

  Also note that invariant (iii) is preserved, because $s$ is not a
  leaf of the tree.

  \smallskip\noindent
  \textit{Case 2:}
  $X= Y_1\cup Y_2$ and $Y_1\cap Y_2\neq\emptyset$.\\
  We let  $\xi_{i+1}(s):=\xi_i(s)$. 

  By posimodularity, either $\kappa(Y_1\setminus Y_2)\le\kappa(Y_1)$
  or $\kappa(Y_2\setminus Y_1)\le\kappa(Y_2)$. If $\kappa(Y_1\setminus
  Y_2)\le\kappa(Y_1)$, we let $\xi_{i+1}(t_1):=Y_1\setminus Y_2$
  and $\xi_{i+1}(t_2):=Y_2$. Otherwise, we let $\xi_{i+1}(t_1):=Y_1$
  and $\xi_{i+1}(t_2):=Y_2\setminus Y_1$. 
  
  Note that in this case we have
  $\kappa(\xi_{i+1}(u))\le \kappa(\xi_{i}(u))$ and $\xi_{i+1}(u)\subseteq\xi_{i}(u)$ for all
  nodes $u$ and either $\xi_{i+1}(t_1)\subset\xi_{i}(t_1)$
  or $\xi_{i+1}(t_2)\subset\xi_{i}(t_2)$.

  \smallskip This completes the description of the construction. To
  see that it terminates, we say that the \emph{total weight} of
  $\xi_i$ is $\sum_{t\in V(T^b)}\kappa(\xi_i(t))$ and the \emph{total
    size} of $\xi_i$ is $\sum_{t\in V(T^b)}|\xi_i(t)|$.  Now observe
  that in each step of the construction either the total weight
  decreases or the total weight stays the same and the total size
  decreases. This proves termination.  

  To see that for all $i$ the triple
  $(T^b,r^b,\xi_i)$ is indeed as partial decomposition, observe first
  that $\xi_i(r^b)=U$, because the root can only occur as the parent
  node $s$ in the construction above, and the set at the
  parent node either stays the same (in Cases~1a and 2) or increases
  (in Case~1b). Moreover, it is easy to check that for all nodes $s'$
  with children $t_1',t_2'$ 
  we have $\xi_i(s')\subseteq\xi_i(t_1')\cup\xi_i(t_2')$. This follows
  immediately from the construction if $s=s'$. If $s'$ is the parent
  of $s=t_i'$, it follows because the set at $s$ can only increase. If
  $s'=t_i$, it follows because the set at $t_i$ can only
  decrease. Otherwise, all the sets at $s',t_1',t_2'$ remain
  unchanged. Note that the invariant (iii) is preserved, because
  leaves can only occur as the child nodes $t_i$ in the construction above, and the set at the
  cild nodes either decrease (in Cases~1a and 2) or stay the same
  (in Case~1b).

  Now we turn the exact binary partial decomposition $(T^b,r^b,\xi_m)$
  into an exact partial decomposition $(T',\tilde\xi')$. Invariants
  (i) and (iii) guarantee that $(T',\tilde\xi')$ has the desired
  properties.
\end{proof}

Let $(T,\tilde\xi)$ be an exact partial decomposition of $\kappa$ over some set
$\CX\subseteq 2^U$.  Observe that we can easily eliminate leaves $u\in
L(T)$ with $\xi(u)=\emptyset$ (we call them \emph{empty leaves}) by
deleting $u$ and contracting the edge from the sibling of~$u$ to its
parent. 
Doing this repeatedly, we can turn the decomposition into a
decomposition of at most the same width over
$\CX\setminus\{\emptyset\}$. 

We say that a tangle $\CT$ \emph{avoids} a
set $\CX\subseteq 2^U$ if $\CT\cap\CX=\emptyset$. Note that, by
\ref{li:t3}, every tangle avoids the set $\CS_U$ of all  singletons.

\begin{lemma}[\mbox{Duality Lemma, \cite{gm10}}]\label{lem:duality}
  Let $\CX\subseteq 2^U$ be closed under taking subsets.  Then there
  is a partial decomposition of width less than $k$ over
  $\CX\cup\CS_U$ if and only if there is no $\kappa$-tangle of order
  $k$ that avoids $\CX$.
\end{lemma}

\begin{proof}
  For the forward direction, let $(T,\tilde\xi)$ a partial decomposition
of $\kappa$ of width less than $k$ over $\CX\cup\CS_U$. Suppose for
contradiction that $\CT$ is a $\kappa$-tangle of order $k$ that avoids
$\CX$. For every edge $st\in E(T)$, we orient $st$ towards $t$ if
$\tilde\xi(s,t)\in\CT$ and towards $s$ if
$\bar{\tilde\xi(s,t)}=\tilde\xi(t,s)\in\CT$. As $\CT$ is a tangle of
order $k$ and $\kappa(\tilde\xi(s,t))<k$ for all $(s,t)\in\vec E(T)$,
every edge gets an orientation. As $T$ is a tree, there is a node
$s\in V(T)$ such that all edges incident with $s$ are oriented towards
$s$. If $s$ is a leaf, then $\xi(s)\in\CT$ and thus
$\xi(s)\not\in\CX\cup\CS_U$ by \ref{li:t3} and because $\CT$ avoids
$\CX$. This contradicts $(T,\tilde\xi)$ being a decomposition over
$\CX\cup\CS_U$. Thus $s$ is an inner node, say, with neighbours
$t_1,t_2,t_3$. Then $\tilde\xi(t_i,s)\in\CT$ and thus
$\tilde\xi(t_1,s)\cap \tilde\xi(t_2,s)\cap
\tilde\xi(t_3,s)\neq\emptyset$. This implies  $\tilde\xi(s,t_1)\cup \tilde\xi(s,t_2)\cup
\tilde\xi(s,t_3)\neq U$, which contradicts $(T,\tilde\xi)$ being a
partial decomposition.

\medskip
The proof of the backward direction is by induction on the number of
nonempty sets $X\subset U$ with $\kappa(X)<k$ such that neither $X$
nor $\bar X$ are in $\CX$.

  For the base case, let us assume that for all $X\subset U$ with
  $\kappa(X)<k$ either $X\in\CX$ or $\bar X\in\CX$. Let
  \[
  \CY=\{\bar X\mid X\in\CX\text{ with }\kappa(X)<k\}.
  \]
  Then
  $\CY$ trivially satisfies the tangle axiom \ref{li:t0}. It satisfies
  \ref{li:t1} by our assumption that either $X\in\CX$ or $\bar
  X\in\CX$ for all $X\subset U$ with
  $\kappa(X)<k$.

  If
  $\CY$ violates \ref{li:t2}, then there are sets $Y_1,Y_2,Y_3\in\CY$
  with $Y_1\cap Y_2\cap Y_3=\emptyset$. We let $T$ be the tree
  with vertex set $V(T)=\{s,t_1,t_2,t_3\}$ and edge set
  $\{st_1,st_2,st_3\}$, and we define $\tilde\xi(t_i,s):=Y_i$ and
  $\tilde\xi(s,t_i):=\bar Y_i\in\CX$. Then $(T,\tilde\xi)$ is a
  partial branch decomposition over $\CX$ of width less than $k$.
  
  So let us assume that $\CY$ satisfies \ref{li:t2}. As there is no
  tangle of order $k$, the set $\CY$ must violate \ref{li:t3}. Thus
  for some $x\in U$ we have $\{x\}\in\CY$ and thus $\bar{\{x\}}\in\CX$ and $\kappa(\bar{\{x\}})=\kappa(\{x\})<k$. 
  We let $T\cong K_2$ be a tree consisting of just one edge, say, $st$, and
  define $\tilde\xi$ by $\tilde\xi(s,t)=\{x\}$,
  $\tilde\xi(s,t)=\bar{\{x\}}$. This yields  a
  partial branch decomposition over $\CX\cup\CS_U$ of width less than $k$.

  \medskip For the inductive step, suppose that there is no partial
  decomposition of $\kappa$ of width less than $k$ over $\CX$. Let $X\subset U$ such that
  $\kappa(X)<k$ and neither $X\in\CX$ nor $\bar X\in\CX$ and such that
  $|X|$ is minimum subject to these conditions. Let
  $\CX^1:=\CX\cup 2^X$ and
  $\CX^2:=\CX\cup 2^{\bar X}$. Then by
  the inductive hypothesis, for $i=1,2$ there is a partial decomposition
  $(T^i,\tilde\xi^i)$ of width less than $k$ over $\CX^i$. 
As there is
  no partial decomposition of width less than $k$ over $\CX$, there
  is a leaf $t^i$ of $T^i$ with $\xi(t^i)\not\in\CX$. 

  Consider $(T^1,\tilde\xi^1)$. By the Exactness Lemma and since
  $\CX^1$ is closed under taking subgraphs, we
  may assume that $(T^1,\tilde\xi^1)$ is exact. This implies that the
  sets $\xi^1(t)$ for the leaves $t\in L(T^1)$ are mutually
  disjoint. Let $t^1$ be a leaf of $T^1$ with
  $X':=\xi^1(t^1)\not\in\CX$. Then $X'\subseteq X$ and $\bar
  X\subseteq\bar X'$, and as $\bar X\not\in\CX$ and $\CX$ is closed
  under taking subsets, it follows that $\bar
  X'\not\in\CX$. By the minimality of $X$, this implies
  $X'=X$. Furthermore, as the decomposition $(T^1,\tilde\xi^1)$ is
  exact, there is exactly one leaf $t^1$ of $T^1$ with
  $\xi^1(t^1)=X$, and for all other leaves $t$ we have
  $\xi^1(t^1)\in\CX$. Let $s^1$ be the neighbour of $t^1\in T^1$.

  Now consider $(T^2,\tilde\xi^2)$. Let $t^2_1,\ldots,t^2_m$ be an
  enumeration of all
  leaves $t^2\in L(T^2)$ with $\xi^2(t^2)\not\in\CX$. Then
  $\xi(t^2_i)\subseteq \bar X$ for all $i\in[m]$. Without loss of
  generality we may assume that $\xi(t^2_i)=\bar X$, because
  increasing a set $\xi^2(t)$ preserves the property of being a
  partial decomposition. (We do not assume $(T^2,\tilde\xi^2)$ to be
  exact.) For every $i\in[m]$, let $s^2_i$ be the
  neighbour of $t^2_i$ in $T^2$.

  To construct a partial decomposition $(T,\tilde\xi)$ of $\kappa$
  over $\CX$, we take $m$ disjoint copies 
  \[
  (T^1_1,\tilde\xi^1_1),\ldots,(T^1_m,\tilde\xi^1_m) 
  \]
  of
  $(T^1,\tilde\xi^1)$. For each node $t\in V(T^1)$, we denote its copy
  in $T^1_i$ by $t_i$. Then for every edge $st\in E(T^1)$ we have
  $\tilde\xi^1_i(s_i,t_i)=\tilde\xi^1(s,t)$. In particular,
  $\tilde\xi^1_i(s^1_i,t^1_i)=\tilde\xi^1(s^1,t^1)=X$. We let $T$ be
  the tree obtained from the disjoint union of
  $T^1_1,\ldots,T^1_m,T^2$ by deleting the nodes $t^1_i,t^2_i$ and
  adding edges $s^1_is^2_i$ for all $i\in[m]$. We define
  $\tilde\xi:V(T)\to 2^U$ by
  \[
  \tilde\xi(s,t):=
  \begin{cases}
    X&\text{if }(s,t)=(s^1_i,s^2_i)\text{ for some }i\in[m],\\
    \bar X&\text{if }(s,t)=(s^2_i,s^1_i)\text{ for some }i\in[m],\\
    \tilde\xi^1_i(s,t)&\text{if }st\in E(T^1_i),\\
    \tilde\xi^2(s,t)&\text{if }st\in E(T^2).
  \end{cases}
  \]
  It is easy to see that $(T,\tilde\xi)$ is a partial decomposition of
  width less than $k$ over $\CX$. This is a contradiction. 
\end{proof}

Since we may assume decompositions to be exact, the Duality Lemma
implies that $\kappa$ has a branch decomposition of width less than
$k$ if and only if there is no $\kappa$-tangle of order $k$, in other
words: the branch width of $\kappa$ is exactly the maximum order of a
$\kappa$-tangle.

The following somewhat surprising fact follows from the existence
claim of Theorem~\ref{theo:candec} (see
Remark~\ref{rem:linear}). However, the fact is needed to construct our
tangle data structure (Theorem~\ref{theo:ds}) and hence to prove the
algorithmic claim of Theorem~\ref{theo:candec}.\footnote{Thus a reader
  interested in a self-contained treatment may skip
  Section~\ref{sec:ds} at first reading and immediately jump to
  Section~\ref{sec:tree}, but ignore all algorithmic parts there, then
  go back to Section~\ref{sec:ds}, and finally to the algorithmic
  parts of Section~\ref{sec:tree}. See Remark~\ref{rem:linear} for
  further details.}

\begin{fact}[\mbox{\cite{gm10}}]\label{fact:lin}
  Let $k\ge 0$, and let $\kappa$ be a a connectivity function on a set
  $U$ of size $|U|=n$. Then there are at most $n$ $\kappa$-tangles of
  order $k$.
\end{fact}

Actually, Corollary~(10.4) of \cite{gm10} only states that there are at most $n$
maximal tangles. But as every maximal tangles contains at most one
tangle of order $k$, the fact follows.

\subsection{Bases}

For disjoint sets~$X,Y\subseteq U$ we define~$
\kappa_{\min}(X,Y):=\min\{\kappa(Z)\mid X\subseteq Z\subseteq
\overline{Y}\}$. Note that for all $X,Y$ the two functions
$X'\mapsto\kappa_{\min}(X',Y)$ and $Y'\mapsto\kappa_{\min}(X,Y')$ are
monotone and submodular.

For sets $Y\subseteq X$, we say that a set $Y$ is \emph{free} in $X$ if
$\kappa_{\min}(Y,\bar X) = \kappa(X)$ and~$|Y| \leq  \kappa(X)$. Let
us remark that if $\kappa(\emptyset)=0$ and
$\kappa(\{x\})=1$ for all $x\in X$, then the function
$Y\mapsto\kappa_{\min}(Y,\bar X)$ is the rank function of a matroid on
the set $X$, and a set $Y$ is free in $X$ if and only if it is a base
for this matroid. 

The following simple and well-known
lemma (see, for example, \cite{oumsey07}) thus generalises the fact
that every matroid, defined by its rank function, has a base.

\begin{lemma}\label{lem:existence:of:free:sets}
  For every~$X\subseteq U$ there is a set~$Y$ that is free in $X$.
\end{lemma}

\begin{proof}
  Let~$Y\subseteq X$ be an inclusion-wise minimal set
with~$\kappa_{\min}(Y,\bar X) = \kappa(X)$. Such a set exists because $\kappa_{\min}(X,\bar X) = \kappa(X)$. We claim
that~$|Y| \leq \kappa(X)$. Suppose otherwise and let~$Y=
\{y_1,\ldots,y_\ell\}$ with~$\ell> \kappa(X)$. Define~$Y_i =
\{y_1,\ldots,y_i\}$ for all~$i\in
[\ell]$. Then~$\kappa_{\min}(Y_i,\bar X)\leq
\kappa_{\min}(Y_{i+1},\bar X)$ since~$\kappa_{\min}$ is monotone in the first
argument. This implies that there exists a~$j\in [\ell-1]$
such that~$\kappa_{\min}(Y_j,\bar X)= \kappa_{\min}(Y_{j+1},\bar X)$.
By the submodularity of $\kappa_{\min}$ in the first argument,
\[
\kappa_{\min}(Y\setminus \{y_{j+1}\},\bar X) +
\kappa_{\min}(Y_{j+1},\bar X) \geq \kappa_{\min}(Y_{j},\bar X)+
\kappa_{\min}(Y,\bar X).
\]
This implies~$\kappa_{\min}(Y\setminus
\{y_{j+1}\},\bar X)= \kappa_{\min}(Y,\bar X) = \kappa(X)$ contradicting the minimality of~$Y$.
\end{proof}

For disjoint sets $B_1,B_2$, we let $\CL(B_1,B_2)$ be the set of all
$X\subseteq U$ such that $B_1$ is free in $X$ and $B_2$ is free in
$\bar X$. If $X\in\CL(B_1,B_2)$, then we say that $(B_1,B_2)$ is a
\emph{base} for $X$. 

\begin{corollary}\label{cor:basis}
  Every set $X\subseteq U$ has a base.
\end{corollary}

For every $X\subseteq U$, we let $\CB(X)$ be
the set of all bases for $X$. Then $(B_1,B_2)\in\CB(X)\iff
X\in\CL(B_1,B_2)$. We let $\CB=\bigcup_{X\subseteq U}\CB(X)$ be the set of all
bases. Obviously, $\CB$ is the set of all pairs
$(B_1,B_2)$ such that $B_1,B_2\subseteq U$ are disjoint and
$\CL(B_1,B_2)\neq\emptyset$. 

\begin{lemma}\label{lem:f1}
   Let $B_1,B_2\subseteq U$ be disjoint and
  $X\in\CL(B_1,B_2)$. Then 
  $
  \kappa(X)=\kappa_{\min}(B_1,B_2).
  $
\end{lemma}

\begin{proof}
  As $B_1\subseteq X\subseteq\bar B_2$, we have
  $\kappa(X)\ge\kappa_{\min}(B_1,B_2)$. Suppose for contradiction that
  $\kappa(X)>\kappa_{\min}(B_1,B_2)$, and let $Z\subseteq U$ such that
  $B_1\subseteq Z\subseteq\bar B_2$ and
  $\kappa(Z)=\kappa_{\min}(B_1,B_2)$. Then $\kappa(Z)<\kappa(X)$, and
  by submodularity this implies $\kappa(X\cap Z)<\kappa(X)$ or
  $\kappa(X\cup Z)<\kappa(X)$. However, we have $B_1\subseteq X\cap
  Z\subseteq X$ and thus $\kappa(X\cap Z)\ge\kappa_{\min}(B_1,\bar
  X)=\kappa(X)$. Similarly, $B_2\subseteq \bar{X}\cap
  \bar{Z}\subseteq \bar X$ and thus $\kappa(X\cup Z)=\kappa(\bar
  X\cap\bar Z)\ge\kappa_{\min}(B_2,X)=\kappa(X)$. This is a contradiction.
\end{proof}

We call $\kappa_{\min}(B_1,B_2)$ the \emph{order} of a base
$(B_1,B_2)\in\CB$. For every $k\ge0$, let $\CB_{\le k}$ denote the set
of all bases of order at most $k$.

\begin{lemma}\label{lem:f2}
  Let $B_1,B_2\subseteq U$ be disjoint. Then $(B_1,B_2)\in\CB$  if and only if
  $\kappa_{\min}(B_1,B_2)\ge \max\{|B_1|,|B_2|\}$. 
\end{lemma}

\begin{proof}
  For the forward direction, let $X\in \CL(B_1,B_2)$. Then
  $\kappa_{\min}(B_1,B_2)=\kappa(X)\ge \max\{|B_1|,|B_2|\}$, where the equality
  follows from Lemma~\ref{lem:f1} and the inequality because $B_1$ is free in
  $X$ and $B_2$ is free in $\bar X$.

  For the backward direction, suppose that
  $\kappa_{\min}(B_1,B_2)\ge \max\{|B_1|,|B_2|\}$, and let $X\subseteq U$ such
  that $B_1\subseteq X\subseteq \bar B_2$ and
  $\kappa(X)=\kappa_{\min}(B_1,B_2)$. Then $\kappa(X)\ge|B_1|$ and
  $\kappa(X)\ge\kappa_{\min}(B_1,\bar X)\ge
  \kappa_{\min}(B_1,B_2)=\kappa(X)$ and thus
  $\kappa(X)=\kappa_{\min}(B_1,\bar X)$. Similarly,
  $\kappa(X)\ge|B_2|$ and $\kappa(X)=\kappa_{\min}(B_2,X)$. Thus
  $(B_1,B_2)$ is a base for $X$.
\end{proof}

\begin{lemma}\label{lem:f3}
  Let $B_1,B_2\subseteq U$ be disjoint and
  $X_1,X_2\in\CL(B_1,B_2)$. Then 
  $X_1\cap X_2,X_1\cup X_2\in\CL(B_1,B_2)$.
\end{lemma}

\begin{proof}
  As $B_1\subseteq X_1\cap X_2,X_1\cup X_2\subseteq \bar B_2$, by
  Lemma~\ref{lem:f1} we have 
  \[
  \kappa(X_1\cap
  X_2),\kappa(X_1\cup
  X_2)\ge\kappa_{\min}(B_1,B_2)=\kappa(X_1)=\kappa(X_2)
  \]
  By submodularity, this implies
  \[
  \kappa_{\min}(B_1,B_2)=\kappa(X_1)=\kappa(X_2)=\kappa(X_1\cap X_2)=\kappa(X_1\cup X_2).
  \]
  By the monotonicity of $\kappa_{\min}$ in the second argument, we have
  \[
  \kappa(X_1\cap X_2)\ge\kappa_{\min}(B_1,\bar{X_1\cap X_2})\ge
  \kappa_{\min}(B_1,B_2)=\kappa(X_1\cap X_2)
  \]
  and thus $ \kappa(X_1\cap X_2)=\kappa_{\min}(B_1,\bar{X_1\cap
    X_2})$. Moreover, $\kappa(X_1\cap X_2)=\kappa(X_1)\ge|B_1|$, and
  thus $B_1$ is free in $X_1\cap X_2$. 
  Similarly,
  \begin{align*}
  \kappa(\bar{X_1\cap X_2})
  \ge\kappa_{\min}(B_2,X_1\cap X_2)
  &\ge
  \kappa_{\min}(B_2,B_1)\\
  &=\kappa_{\min}(B_1,B_2)=\kappa(X_1\cap X_2)=\kappa(\bar{X_1\cap X_2})
  \end{align*}
  and thus $ \kappa(\bar{X_1\cap X_2})=\kappa_{\min}(B_2,X_1\cap
    X_2)$. Moreover, $\kappa(\bar{X_1\cap X_2})=\kappa(X_1\cap X_2)=\kappa(X_2)\ge|B_2|$, and
  thus $B_2$ is free in $\bar{X_1\cap X_2}$. 

  Hence $X_1\cap X_2\in\CL(B_1,B_2)$. By symmetry, $X_1\cup
  X_2\in\CL(B_1,B_2)$.
\end{proof}

The lemma shows that $\CL(B_1,B_2)$ has a lattice structure---hence the
letter $\CL$. In particular, the lemma has the following consequence.

\begin{corollary}
  For every base $(B_1,B_2)\in\CB$ there is a unique minimal
  $X_{\bot}(B_1,B_2)\in\CL(B_1,B_2)$  and a
  unique maximal $X_{\top}(B_1,B_2)\in\CL(B_1,B_2)$.
\end{corollary}

Note that in contrast to the leftmost separations defined earlier, the element~$X_{\bot}(B_1,B_2)$ is inclusion-wise minimal among all sets in~$\CL(B_1,B_2)$ and not just among those of minimum rank.

\subsection{Computing with tangles}

Algorithms expecting a set function $\kappa:2^U\to\NN$ as input are
given the ground set $U$ as actual input (say, as a list of objects),
and they are given an oracle that returns for~$X\subseteq U$ the value
of~$\kappa(X)$. The running time of such algorithms is measured in
terms of the size $|U|$ of the ground set. We assume this computation
model whenever we say that an algorithm is given \emph{oracle access} to a
set function $\kappa$.

An important fact underlying most of our algorithms is that, under this model of
computation, submodular functions can be efficiently minimised.

\begin{fact}[Iwata, Fleischer, Fujishige~\cite{iwaflefuj01}, Schrijver~\cite{schrijver00}]
  There is a polynomial time algorithm that, given oracle access to a
  submodular function $\lambda:2^U\to\NN$, returns a set $X\subseteq U$ that
  minimises $\lambda$.
\end{fact}

Observe that this implies that, given oracle access to a connectivity
function $\kappa$, the function~$\kappa_{\min}$  can also be evaluated in
polynomial time. In fact, for given arguments $X,Y$ we can compute a
$Z$ such that $X\subseteq Z\subseteq\bar Y$ and
$\kappa(Z)=\kappa_{\min}(X,Y)$ in polynomial time.

A \emph{membership oracle} for a family $\CS\subseteq 2^U$ of subsets
of a ground set $U$ is an oracle that when queried with a
set~$X\subseteq U$ determines whether~$X\in \mathcal{S}$. Most often,
we will use such membership oracles for tangles; in the following
lemma we use them for the union of a family of tangles.

\begin{figure}
  \centering
  \begin{tikzpicture}[thick, scale = 1.5]
  \draw (2,1.5) .. controls (2.5,1) and (2.5,-1.5) .. (2.25,-2);
  \node at (1.85,1.25) {$X'$};
  \draw (3,1.5) .. controls (3.5,1) and (3.5,-1.5) .. (3.5,-2);
  \node at (2.85,1.25) {$X$};
  \draw (1,1) .. controls (2,0.75) and (3.5,-1.25) .. (3.5,-2);
  \node at (1,0.75) {$X^{\ast}$};
  \draw (4.5,0.25) .. controls (4,0) and (4,-1) .. (4.5,-1.25);
  \node at (4.5,-0.5) {$\overline{Y_2}$};
  \draw (0.75,0.25) .. controls (1.25,0) and (1.25,-1) .. (0.75,-1.25);
  \node at (0.75,-0.5) {${Y_1}$};
  \draw (0.75,-1) .. controls (1.25,-1) and (1.75,-1.5) .. (1.75,-2);
  \node at (1.25,-1.5) {${Z}$};
  \node (v2) at (2.875,0.5) {$x$};
  \draw[fill]  (2.75,0.25) node (v1) {} ellipse (0.05 and 0.05);
  \end{tikzpicture}
  \caption{Proof of Lemma~\ref{lem:minimal:separators:are:computable}}
  \label{fig:comp-min-sep}
\end{figure}
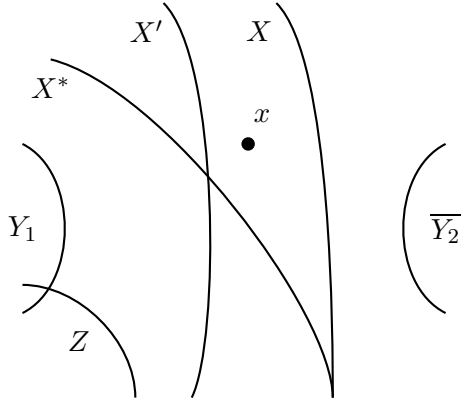

\begin{lemma}\label{lem:minimal:separators:are:computable}
Let~$k\ge 1$.
There is a polynomial time algorithm that, given oracle access to a
connectivity function $\kappa:2^U\to\NN$, a membership oracle for set
$\CS\subseteq 2^U$ such that $\CS=\CT_1\cup\dots\cup\CT_n$ is a union
of $\kappa$-tangles of order $k$, and sets $Y_1,Y_2\subseteq U$,
\begin{itemize}
  \item
    either computes a set $X'$ satisfying
\begin{enumerate}
\item $X'\in \mathcal{T}_1 \cup\dots\cup \mathcal{T}_n$\label{item:in:tangle} and 
\item $Y_1\subseteq X' \subseteq Y_2$ \label{item:subset}
\end{enumerate}
such that $X'$ is inclusion-wise minimal among all sets with these two
properties 
\item
or determines that no such set exists.
\end{itemize}
\end{lemma}

\begin{proof}
Starting by initially setting~$X= Y_2$, 
it suffices to find repeatedly a proper subset~$X'\subset X$ satisfying Properties~\ref{item:in:tangle} and~\ref{item:subset} or determine that no such set exists.

Assume there is some proper subset~$X^*\subset X$ satisfying the properties. Let~$x$ be an element of~$X\setminus X^*$ and let~$Z$ be a free subset of~$X^*$. In particular~$\kappa_{\min}(Z,\overline{X^*}) = \kappa(X^*)$ and~$|Z| \leq \kappa(X^*)$. Such a set exists by Lemma~\ref{lem:existence:of:free:sets}.

Since~$|Z|\leq \kappa(X^*)\leq k-1$, the number of choices for~$x$ and~$Z$ are polynomially bounded, and by iterating over all choices we can guess~$x$ and~$Z$.
We compute an inclusion-wise maximal set~$X'$ with~$\kappa(X')\leq
k-1$ such that~$Z\cup Y_1\subseteq X'\subseteq X\setminus \{x\}$ (see Figure~\ref{fig:comp-min-sep}). Such
a set $X'$ can be computed by starting with $X':=Z\cup Y_1$ and then
iteratively adding elements to $X'$ as long as
$\kappa_{\min}(X',\bar{X\setminus\{x\}})\leq k-1$.
Note that~$Z\subseteq X'\cap X^*\subseteq X^*$ and thus~$\kappa(X'\cap X^*)\geq \kappa_{\min}(Z,\overline{X^*}) = \kappa(X^*)$. By submodularity~$\kappa(X^*)+ \kappa(X') \geq  \kappa(X'\cap X^*) + \kappa(X'\cup X^*)$ which implies that~$\kappa(X'\cup X^*) \leq \kappa(X')$ and thus~$X^*\subseteq X'$, since~$X'$ is chosen to be maximal.
We conclude that~$X'\in \mathcal{T}_1 \cup\dots\cup \mathcal{T}_n$. For a correct choice of~$Z$ and~$x$, we can use the oracle for~$\mathcal{T}_1 \cup\dots\cup \mathcal{T}_n$ to verify that~$X'\in \mathcal{T}_1 \cup\dots\cup \mathcal{T}_n$.
\end{proof}

The lemma in particular applies when we are given only one tangle by a
membership oracle.

\begin{lemma}\label{lem:f4}
  Let $\CT$ be a tangle of order $k$ and $B=(B_1,B_2)$ a base such that
  $\CT\cap\CL(B)\neq\emptyset$. Then there is a unique inclusion-wise
  minimal element $X_{\bot}(\CT,B)$ in $\CT\cap\CL(B)$.

  Furthermore, for every fixed $k$ there is a polynomial time
  algorithm that, given oracle access to $\kappa$, a membership oracle
  for $\CT$, and a base
  $B$, decides if $\CT\cap\CL(B)\neq\emptyset$  and, in the affirmative case, computes
  $X_{\bot}(\CT,B)$.
\end{lemma}

\begin{proof}
  For the uniqueness, it is sufficient to prove that
  $\CT\cap\CL(B)$ is closed under taking intersections. So let
  $X_1,X_2\in\CT\cap\CL(B)$. Then by Lemma~\ref{lem:f3}, $X_1\cap
  X_2\in\CL(B)$ and thus by Lemma~\ref{lem:f1},
  $\kappa(X_1\cap X_2)=\kappa(X_1)=\kappa(X_2)$. As 
  $
  \bar{X_1\cap X_2}\cap X_1\cap X_2=\emptyset,
  $
  we have $X_1\cap X_2\in\CT$.

  The algorithmic claim follows by applying  Lemma~\ref{lem:minimal:separators:are:computable} with~$n= 1$,~$\mathcal{T}_1 = \mathcal{T}$,~$Y_1 = B_1$ and~$Y_2 = B_2$.
\end{proof}

\section{A Data Structure for Tangles}
\label{sec:ds}

The following crucial lemma will enable us to iteratively extend our
tangle data structure.

\begin{lemma}\label{lem:given:Xs:compute:tangle:avoiding}
  There is a polynomial time algorithm that, given oracle access to
  $\kappa$, a membership oracle for a $\kappa$-tangle $\CT_0$ of order at most $k$, and sets
  $X_1,\ldots,X_n\subseteq U$ with $\kappa(X_i)\le k$, decides if there
  is a $\kappa$-tangle $\CT\supseteq\CT_0$ of order $k+1$ that avoids
  $X_1,\ldots,X_n$.
\end{lemma}

\begin{proof}
  Observe that there
  is a $\kappa$-tangle $\CT\supseteq\CT_0$ of order $k+1$ that avoids
  $X_1,\ldots,X_n$ if and only if there is a $\kappa$-tangle of order
  $k+1$ that avoids
  \[
  \CX:=\underbrace{\{X\mid X\subseteq \bar Y\text{ for some
    }Y\in\CT_0\}}_{=:\CX_0}
  \;\cup\;
  \underbrace{\{X\mid X\subseteq X_i\text{ for some
    }i\in[n]\}}_{=:\CX_1}
  \;\cup\;\CS_U.
  \]
  By the Duality Lemma, this is the case if and only if there is no
  partial decomposition of $\kappa$ of width less than $k+1$ over
  $\CX$.
  Our algorithm will test whether such a partial decomposition exists.

  For a set $Y\subseteq U$, we say that a \emph{$Y$-decomposition} is
  a partial decomposition $(T,\tilde\xi)$ of $\kappa$ of width less
  than $k+1$ that has a leaf $t_Y$, which we call the
  \emph{$\bar Y$}-leaf, such that $\xi(t_Y)=\bar Y$ and $\xi(t)\in\CX$
  for all $t\in L(T)\setminus\{t_Y\}$. We call a set $Y$
  \emph{decomposable} if there is a $Y$-decomposition.

  \begin{claim}
    Let $X,Y,Z\subseteq U$ such that $X$ and $Y$ are decomposable and
    $Z\subseteq X\cup Y$ and 
    $\kappa(Z)\le k$. Then $Z$ is decomposable.

    \proof
    Let $(T_X,\tilde\xi_X)$ be an $X$-decomposition and
    $(T_Y,\tilde\xi_Y)$ a $Y$-decomposition. Let $t_X$ be the $\bar X$-leaf of
    $T_X$, and let $t_Y$ be the $\bar Y$-leaf of
    $T_Y$. We form a new tree $T$ by taking
    the disjoint union of $T_X$ and $T_Y$, identifying the two leaves
    $t_X$ and $t_Y$, and adding a fresh node $t_Z$ and an edge between
    $t_Z$ and the
    node $s_Z$ corresponding to $t_X$ and $t_Y$. We define $\tilde\xi:\vec
    E(T)\to 2^U$ by 
    \[
    \tilde\xi(s,t):=
    \begin{cases}
      \tilde\xi_X(s,t)&\text{if }(s,t)\in\vec E(T_X),\\
      \tilde\xi_Y(s,t)&\text{if }(s,t)\in\vec E(T_Y),\\
      Z&\text{if }(s,t)=(t_Z,s_Z),\\
      \bar{Z}&\text{if }(s,t)=(s_Z,t_Z).
    \end{cases}
    \]
    Then $(T,\tilde\xi)$ is a partial decomposition of $\kappa$,
    because $X\cup Y\cup\bar Z=U$. Its width is at most $k$ because
    $\kappa(Z)\le k$ and the width of both $(T_X,\tilde\xi_X)$ and
    $(T_Y,\tilde\xi_Y)$ is at most $k$. We have $\xi(t_Z)=\bar Z$, and for all leaves
    $t\in L(T)\setminus\{t_Z\}$, either $t\in L(T_X)\setminus\{t_X\}$
    and $\xi(t)=\xi_X(t)\in\CX$ or $t\in L(T_Y)\setminus\{t_Y\}$ and
    $\xi(t)=\xi_Y(t)\in\CX$. Thus $(T,\tilde\xi)$ is a
    $Z$-decomposition.
    \uend
  \end{claim}

  \begin{claim}[2]
    Let $X,Y\subseteq U$ such that $X$ and $Y$ are decomposable and
    $X\cup Y=U$. Then there is a partial decomposition of $\kappa$ of
    width at most $k$ over $\CX$.

    \proof
    This follows from Claim~1 by observing that an $U$-decomposition
    is a partial decomposition of $\kappa$ over $\CX$ of width at most
    $k$. Here we use the facts that $\bar U=\emptyset\in \CX$ and that
    $\kappa(\emptyset)=0\le k$.
    \uend
  \end{claim}

  Our algorithm iteratively computes mappings $\mu_1,\ldots,\mu_m:\CB_{\le k}\to
  2^U$ such that for all $B=(B_1,B_2)\in\CB_{\le k}$, either
  $\mu_i(B)=\emptyset$ or $\mu_i(B)\in\CL(B)$ and 
  $\mu_i(B)$ is decomposable.

  To define $\mu_1$, for every base $B=(B_1,B_2)\in\CB_{\le k}$, we
  first check if $\CT_0\cap\CL(B_2,B_1)\neq\emptyset$ (note the
  reversed order in the base) and, if it
    is,  compute $X_\bot(\CT_0,(B_2,B_1))$ (see Lemma~\ref{lem:f4}). 
  \begin{itemize}
  \item If $\CT_0\cap \CL(B_2,B_1)\neq\emptyset$, we let
    $Y_0:=\bar{X_\bot(\CT_0,(B_2,B_1))}$. Observe that $Y_0$ is the unique
    inclusion-wise maximal element in $\CL(B)\cap \CX_0$.  Otherwise, we
    let $Y_0:=\emptyset$.
  \item If $\CL(B)\cap\CX_1\neq\emptyset$, then we let $Y_1$ be the
    union of all elements of $\CL(B)\cap\CX_1$. Otherwise, we let
    $Y_1:=\emptyset$.
  \item If $\CL(B)\cap\CS_U\neq\emptyset$, then we let $Y_2$ be the
    union of all elements of $\CL(B)\cap\CS_U$. Otherwise, we let
    $Y_2:=\emptyset$.
  \end{itemize}
  We let $\mu_1(B):=Y_0\cup Y_1\cup Y_2$. Observing that the elements
  of~$\CX$ are trivially decomposable, it follows from Claim~1 and the
  fact that~$\CL(B)$ is a lattice that if $\mu_1(B)\neq\emptyset$ then
  it is decomposable.

  Now suppose that $\mu_i$ is defined. Consider the following
  condition.
  \begin{itemize}
  \item[($\star$)] For all
    $B,C,D\in\CB_{\le k}$ and
    $Y\in\CL(B)$, if $Y\subseteq\mu_i(C)\cup\mu_i(D)$, then $Y\subseteq\mu_i(B)$.
  \end{itemize}
  If ($\star$) is satisfied, then we let $m:=i$, and the
  constructions stops.

  Otherwise, we find $B,C,D\in\CB_{\le k}$ and
  $Y\in\CL(B)$ such that $Y\subseteq\mu_i(C)\cup\mu_i(D)$ and $Y\not\subseteq\mu_i(B)$. (Oum and Seymour \cite{oumsey07}
  explain how to do this in polynomial time: for all $B,C,D\in\CB_{\le
    k}$ and for all $y\in\bar{\mu_i(B)}$ 
  we test if
  there is a $Y\in\CL(B)$ with $\{y\}\subseteq
  Y\subseteq\mu_i(C)\cup\mu_i(D)$.) 

  We let $\mu_{i+1}(B):=\mu_i(B)\cup
  Y$. It follows from Claim~1 that $\mu_{i+1}(B)$ is decomposable. For
  all $B'\in\CB_{\le k}\setminus\{B\}$, we let $\mu_{i+1}(B'):=\mu_i(B')$.
  This completes the description of the construction.
  
Since in each step of the construction we strictly increase the set
$\mu_i(B)$ for some base $B$ and since the number of bases is
polynomially bounded in $|U|$, the number $m$ of steps of the
construction is polynomially bounded as well.
  Let $\mu:=\mu_m$. Note that $\mu$ satisfies ($\star$).

  \begin{claim}[3]
    The following are equivalent.
    \begin{enumerate}
    \item There is a partial decomposition of $\kappa$ of
    width at most $k$ over $\CX$.
    \item There are $B,C\in\CB_{\le k}$ such that $\mu(B)\cup\mu(C)=U$.
    \end{enumerate}

    \proof
    The implication ``(2)$\implies$(1)'' follows immediately from
    Claim~2.

    To prove ``(1)$\implies$(2)'', let $(T,\tilde\xi)$ be an exact
    partial decomposition of $\kappa$ of width at most $k$ over
    $\CX$. Let
    \[
    \CY:=\{Y\mid \kappa(Y)\le k\text{ and }Y\subseteq\mu(B)\text{ for some }B\in\CB_{\le k}\}.
    \]
    By induction on $T$, we shall prove that $\tilde\xi(s,t)\in\CY$
    for all $(s,t)\in\vec E(T)$.

    In the base step, we consider $(s,t)\in\vec E(T)$ where $t$ is a
    leaf. Then $\xi(t)\in\CX$. We let $B$ be a base for
    $\xi(t)$. Then $\xi(t)\subseteq\mu_1(B)\subseteq\mu(B)$ and thus
    $\tilde\xi(s,t)=\xi(t)\in\CY$.

    For the inductive step, consider an edge $(s,t) \in\vec E(T)$ such
    that $t$ is an inner node and such that
    $\tilde\xi(t,u),\tilde\xi(t,v)\in\CY$ for the other two
    neighbours $u,v$ of $t$. Let $C,D\in\CB_{\le k}$ such that
    $\tilde\xi(t,u)\subseteq\mu(C)$ and $\tilde\xi(t,v)\subseteq\mu(D)$.
    Let $B=(B_1,B_2)$ be a base for $Y:=\tilde\xi(s,t)$, and note
    that
    $Y=\tilde\xi(t,u)\cup\tilde\xi(t,v)$, because $(T,\tilde\xi)$ is
    an exact partial decomposition. Hence $Y\subseteq\mu(C)\cup\mu(D)$. By ($\star$), we have $Y\subseteq\mu(B)$ and thus $Y\in\CY$.

    Thus $\tilde\xi(s,t)\in\CY$ for all $(s,t)\in\vec E(T)$. Consider
    an arbitrary edge $st\in E(T)$. Let $B,C\in\CB_{\le k}$ such
    that $\tilde\xi(s,t)\subseteq\mu(B)$ and
    $\bar{\tilde\xi(s,t)}=\tilde\xi(t,s)\subseteq\mu(C)$. Then
    $\mu(B)\cup\mu(C)=U$.  \uend
  \end{claim}

  Since Condition (2) of Claim~3 can be tested in polynomial time, this completes
  the proof.
\end{proof}

Observe that the following result due to Oum and Seymour \cite{oumsey07} follows by
applying the theorem with $\CT_0$ being the empty tangle (of order
$0$) and $n=0$. In fact, our proof of the lemma builds on Oum and
Seymour's proof.

\begin{corollary}[Oum and Seymour~\cite{oumsey07}]\label{cor:oumsey}
  For every $k$ there is a polynomial time algorithm deciding whether
  there is a $\kappa$-tangle of order $k$.
\end{corollary}

A \emph{comprehensive tangle data structure} of order~$k$ for a connectivity function~$\kappa$ over a set~$U$ is a data structure~$\mathcal{D}$ with procedures~$\order_\mathcal{D}, \size_\mathcal{D}, \mathcal{T}_\mathcal{D}, \tangorder_\mathcal{D}, \trunc_\mathcal{D}$, $\sep_\mathcal{D}$, and~$\find_\mathcal{D}$ that provide the following functionalities.

\begin{enumerate}
\item The function~$\order_\mathcal{D}()$ returns the fixed integer~$k$.
\item For~$\ell\in \{1,\ldots,k\}$ the function~$\size_\mathcal{D}(\ell)$ returns the number of~$\kappa$-tangles of
  order at most~$\ell$. We denote the number of~$\kappa$-tangles of
    order at most~$k$ by~$|\mathcal{D}|$.
\item For each~$i\in \{1,\ldots, |\mathcal{D}|\}$ the
  function~$\mathcal{T}_\mathcal{D}(i,\cdot)\colon 2^U \rightarrow
  \{0,1\}$ is a tangle~$\mathcal{T}_i$ of order
  at most~$k$, (i.e., the function call~$\mathcal{T}_\mathcal{D}(i,X)$
  determines whether~$X\in  \mathcal{T}_i$).

  We call $i$ the \emph{index} of the tangle $\CT_i$ within the data
  structure.
\item For~$i\in \{1,\ldots,|\mathcal{D}|\}$ the function~$\tangorder_\mathcal{D}(i)$ returns~$\ord(\mathcal{T}_i)$. %
\item For~$i\in\{1,\ldots,|\mathcal{D}|\}$ and~$\ell \leq \ord(\CT_i)$ the call~$\trunc_\mathcal{D}(i,\ell)$ returns an integer~$j$ such that~$\mathcal{T}_j$ is the truncation of~$\mathcal{T}_i$ to order~$\ell$. If~$\ell >\ord(\CT_i)$ the function returns~$i$.
\item For distinct~$i, j \in \{1,\ldots,|\mathcal{D}|\}$ the call~$\sep_\mathcal{D}(i,j)$ outputs a set~$X\subseteq U$
  such that~$X$ is the leftmost
  minimum~$(\mathcal{T}_i,\mathcal{T}_j)$-separation
(see Lemma~\ref{lem:lm-tansep})
or states that no such set exists (in which case one of the tangles is a truncation of the other).
\item Given an integer~$\ell \in\{0,\ldots,k\}$ and some tangle~$\mathcal{T}'$ of order~$\ell$ (via a
  membership oracle) the function~$\find_\mathcal{D}(\ell,\mathcal{T'})$,
  returns the index of $\CT'$, that is, the unique integer~$i \in \{1,\ldots,|\mathcal{D}|\}$ such that~$\ord(\mathcal{T}_i) = \ell$ and~$\mathcal{T}' = \mathcal{T}_i$.
\end{enumerate}

\begin{remark}
Note that, in accordance with Remark~\ref{rem:tangle-order1}, a tangle data structure considers tangles as distinct even if they only differ in their order and not as a function~$2^U \rightarrow \{0,1\}$. Recall that a tangle of order~$k$ agrees with its truncation to order~$k-1$ as a function~$2^U \rightarrow \{0,1\}$ if and only if there is no $X\subseteq U$ of order $\kappa(X)=k-1$. Due to the existence of free sets (Lemma~\ref{lem:existence:of:free:sets}) it is possible to check in polynomial time whether such an~$X$ exists and thus determine which tangles coincide.
\end{remark}

We say a comprehensive tangle data structure~$\mathcal{D}$ is \emph{efficient} if all functions~$\order_\mathcal{D}$, $\size_\mathcal{D}$, $\mathcal{T}_\mathcal{D}$, $\tangorder_\mathcal{D}$, $\trunc_\mathcal{D}$, $\sep_\mathcal{D}$, and~$\find_\mathcal{D}$ can be evaluated in polynomial time.

\begin{theorem}\label{theo:ds}
For every constant~$k$ there is a polynomial time algorithm that,
given oracle access to a connectivity function $\kappa$, computes an efficient
comprehensive tangle data  structure of order~$k$.
\end{theorem}

\begin{proof}
Let~$k$ be an integer and let~$\kappa \colon 2^U\to\NN$ be a connectivity function. Note that for~$k=0$ it is trivial to construct an efficient comprehensive tangle data structure. 
We suppose by induction that we  have already constructed an efficient comprehensive tangle data  structure of order~$k-1$ in polynomial time.
We first verify that there exists some tangle of order~$k$ (using
Corollary~\ref{cor:oumsey}). If this is not the case then the comprehensive tangle data structure of order~$k-1$ is already a comprehensive tangle data structure of order~$k$. 
Otherwise we proceed as follows.

We compute a binary rooted tree~$T$ and a function~$S$ assigning to every edge~$(s,t)\in E(T)$ that is pointing away from the root a subset~$S(s,t)\subseteq U$  with~$\kappa(S(s,t))< k$ such that the following properties hold.

\begin{enumerate}
\item If~$s$ is a vertex of~$T$ with children~$t_1$ and~$t_2$ then~$S(s,t_1) = \overline{S(s,t_2)}$.\label{item:sets:separate:tangles}
\item  For each path~$p_1,\ldots,p_m$ form the root to a leaf 
there is a tangle~$\mathcal{T}$ of order~$k$ such that $S(p_1,p_2),S(p_2,p_3),\ldots,S(p_{m-1},p_{m}) \in \mathcal{T}$.\label{item:leaves:correspond:to:tangles}
\item For each path~$p_1,\ldots,p_m$ form the root to a leaf there is at most one tangle~$\mathcal{T}$ of order~$k$ satisfying property~\ref{item:leaves:correspond:to:tangles}.\label{item:leaves:correspond:to:at:most:one:tangles}
\end{enumerate}

  \begin{claim}A tree~$T$ satisfying properties~\ref{item:sets:separate:tangles}--\ref{item:leaves:correspond:to:at:most:one:tangles}  can be computed in polynomial time.
  
  \proof
We construct the tree iteratively from subtrees which satisfy properties~\ref{item:sets:separate:tangles} and~\ref{item:leaves:correspond:to:tangles} but not necessarily property~\ref{item:leaves:correspond:to:at:most:one:tangles}.
We say a tangle corresponds to a leaf~$u$ if it satisfies property~\ref{item:leaves:correspond:to:tangles} for the path from the root to~$u$. Note that properties~\ref{item:sets:separate:tangles} and~\ref{item:leaves:correspond:to:tangles} imply that each tangle corresponds to exactly one leaf and that property~\ref{item:leaves:correspond:to:at:most:one:tangles} implies that different tangles correspond to different leaves.

We start with the tree~$T_0$ that only contains one vertex~$r$. It satisfies properties~\ref{item:sets:separate:tangles} and~\ref{item:leaves:correspond:to:tangles}. Suppose we have constructed a tree~$T_i$ satisfying properties~\ref{item:sets:separate:tangles} and~\ref{item:leaves:correspond:to:tangles}. To construct~$T_{i+1}$ it suffices to find a leaf~$u$ and a set~$X$ with~$\kappa(X)<k$ such that there are tangles~$\mathcal{T}_1$ and~$\mathcal{T}_2$ of order at most~$k$ both corresponding to~$u$ with~$X\in \mathcal{T}_1$ and~$X \notin \mathcal{T}_2$.

Note that for each candidate set~$X$ and each
leaf~$u$ we can determine by Lemma~\ref{lem:given:Xs:compute:tangle:avoiding} in polynomial time whether there are two tangles corresponding to~$u$ that are separated by~$X$. It thus suffices for us to compute a set of candidates for~$X$ among which there is an adequate separator. For a leaf~$u$ we proceed as follows. Let~$X_1 = S(p_1,p_2),\ldots,X_{m-1} = S(p_{m-1},p_{m})$ where~$p_1,\ldots,p_m$ is the path from the root to~$u$. (The collection of sets~$X_i$ is empty if the tree has only one node.)
Let~$\mathcal{T}_1,\ldots,\mathcal{T}_n$ be the tangles of order~$k$ which correspond to~$u$. These are exactly the tangles of order~$k$ avoiding~$\overline{X_1},\ldots,\overline{X}_{m-1}$.
Lemma~\ref{lem:given:Xs:compute:tangle:avoiding} provides a membership
oracle to test for a set~$X$ whether~$X\in
\mathcal{T}_1\cup\ldots\cup\mathcal{T}_n$: to check
whether~$X\in \mathcal{T}_1\cup\ldots\cup\mathcal{T}_n$ we test, using
Lemma~\ref{lem:given:Xs:compute:tangle:avoiding} for
$\CT_0:=\emptyset$ (the unique tangle of order $0$), whether there is a tangle of
order~$k$ avoiding~$\overline{X_1},\ldots,\overline{X}_{m-1}$
and~$\bar X$. We have~$X\in  \mathcal{T}_1\cup\ldots\cup\mathcal{T}_n$ if and only if the answer is affirmative for some tangle~$\mathcal{T}$.

For each base~$B\in \CB_{\le k-1}$, we compute, if it exists, an inclusion-wise minimal set~$X^*(B)\in \CL(B)$ (implying~$\kappa(X^*(B))< k$) such that there is a tangle~$\mathcal{T}_i$ (with~$i\in \{1,\ldots,n\}$) containing~$X^*(B)$. Such a set can be computed in polynomial time by Lemma~\ref{lem:minimal:separators:are:computable}.

To prove the claim, it suffices now to show that if there are two
tangles corresponding to~$u$ then there is a base~$B\in \CB_{\le
  k-1}$ for which~$X^*(B)$ exists and separates two tangles
corresponding to~$u$. Let~$\mathcal{T}_j$ and~$\mathcal{T}_k$ be
two such tangles and let~$\widetilde{X}$ be
a~$(\mathcal{T}_j,\mathcal{T}_k)$-separation. Let~$B$ be a base
for~$\widetilde{X}$.  By construction, there is a
tangle~$\mathcal{T}_i$ with~$X^*(B)\in \mathcal{T}_i$. If~$X^*(B)
\notin \mathcal{T}_j$ or $X^*(B) \notin \mathcal{T}_k$ then~$X^*(B)$
separates two tangles corresponding to~$u$ (either $\CT_i$ and
$\CT_j$ or $\CT_i$ and $\CT_k$). So suppose~$X^*(B) \in
\mathcal{T}_j$ and~$X^*(B) \in \mathcal{T}_k$. Then~$X^*(B)
\not\subseteq \widetilde{X}$, because  $\widetilde{X} \notin
\mathcal{T}_k$. Thus $X^*(B) \cap \widetilde{X}\subsetneq X^*(B)$. We
have $X^*(B) \cap \widetilde{X}\in\CL(B)$, because $\CL(B)$ is a
lattice.
Since~$X^*(B)$ is inclusion-wise minimal in $\CL(B)\cap(\CT_1\cup\ldots\cup\CT_t)$, this implies
that~$X^*(B) \cap \widetilde{X} \notin \mathcal{T}_j$.  
Since $\tilde X\in \CL(B)$ and all elements of $\CL(B)$
have the same order, we have $\kappa(X^*(B) \cap \widetilde{X})=\kappa(\tilde
X)\le k-1$. Thus $\overline{X^*(B) \cap \widetilde{X}} \in
\mathcal{T}_j$. However the sets~$X^*(B)$, ${\widetilde{X}}$, and
$\overline{X^*(B) \cap \widetilde{X}}$ have an empty intersection so they
cannot all be contained in~$\mathcal{T}_j$, yielding a
contradiction.

Repeating the construction we obtain a sequence of
trees~$T_0,T_1,\ldots$. Since by Fact~\ref{fact:lin} there is only a
linear number of tangles of order at most $k$ and each new tree distinguishes more tangles than the one before, after a linear number of steps we obtain a tree satisfying properties~~\ref{item:sets:separate:tangles}--\ref{item:leaves:correspond:to:at:most:one:tangles}.\uend
\end{claim}
Let~$\mathcal{D}_{k-1}$ be an efficient comprehensive tangle data structure of order~$k-1$. We argue that using~$\mathcal{D}_{k-1}$ together with a tree~$T$ satisfying properties~\ref{item:sets:separate:tangles}--\ref{item:leaves:correspond:to:at:most:one:tangles}  we obtain a comprehensive tangle data structure of order~$k$.

Recall that the tangles of order~$k$ are in one to one correspondence to the leaves of the tree~$T$. Let~$u_1,\ldots,u_n$ be an enumeration of the leaves of~$T$.
\begin{enumerate}
\item The function~$\order_{\mathcal{D}}()$ simply returns the integer~$k$.
\item The function~$\size_{\mathcal{D}}(\ell)$
  returns~$\size_{\mathcal{D}_{k-1}}(\ell)$ if~$\ell<k$
  and~$\size_{\mathcal{D}_{k-1}}(\ell-1) + n$ if~$\ell =k$
  (where~$n=|L(T)|$ is the number of leaves of~$T$).
\item To evaluate~$\mathcal{T}_\mathcal{D}(i,X)$ we
  return~$\mathcal{T}_{\mathcal{D}_{k-1}}(i,X)$ in
  case~$i\leq |\mathcal{D}_{k-1}|$. Otherwise we suppose
  that~$X_1 = S(p_1,p_2),\ldots,X_{m-1} = S(p_{m-1},p_{m})$
  where~$p_1,\ldots,p_m$ is the path from the root to
  leaf~$u_{i-|\mathcal{D}_{k-1}|}$. We use
  Lemma~\ref{lem:given:Xs:compute:tangle:avoiding} to determine
  whether there exists a tangle
  avoiding~$\overline{X_1},\ldots,\overline{X}_{m-1}$
  and~$\overline{X}$. If this is the case we return~$\mathbf{1}$,
  otherwise we return~$\mathbf{0}$.
\item For~$i\in \{1,\ldots,|\mathcal{D}|\}$ the function~$\tangorder_\mathcal{D}(i)$ returns~$\tangorder_{\mathcal{D}_{k-1}}(i)$ if~$i\leq |\mathcal{D}_{k-1}|$. Otherwise, it returns~$k$.
\item To determine~$\trunc_\mathcal{D}(i,\ell)$, if~$\ell = k$ we
  return~$i$. Assuming otherwise, if~$i\leq |\mathcal{D}_{k-1}|$ we
  use~$\trunc_{\mathcal{D}_{k-1}}(i,\ell)$. If neither of these cases
  happens, we modify the function~$\mathcal{T}_\mathcal{D}(i,X)$ to
  return~$\mathbf{0}$ for all~$X$ with~$\kappa(X) \geq \ell-1$. This
  provides us with an oracle for a tangle~$\mathcal{T}$ that is
  the~$\ell-1$ truncation of~$\mathcal{T}_i$. We can then
  use~$\find_{\mathcal{D}_{k-1}}(k-1,\mathcal{T})$ to determine the
  index of this truncation.
\item To compute~$\sep_\mathcal{D}(i,j)$, if~$i\leq |\mathcal{D}_{k-1}|$ or~$j\leq |\mathcal{D}_{k-1}|$ then we can use the truncation and  simply  compute~$\sep_{\mathcal{D}_{k-1}}(\trunc(i,k-1),\trunc(j,k-1))$.
Otherwise we let~$s$ be the smallest common ancestor of~$u_{i-|\mathcal{D}_{k-1}|}$ and~$u_{j-|\mathcal{D}_{k-1}|}$ in~$T$. Let~$X$ be the set~$S(s,t)$ where~$t$ is the child of~$s$ on the path from~$s$ to~$u_{i-|\mathcal{D}_{k-1}|}$. Using Lemma~\ref{lem:minimal:separators:are:computable} we compute and return the minimal set~$X'\subseteq X$ such that~$X'\in \mathcal{T}_{i}$.
\item To determine $\find_\mathcal{D}(\ell,\mathcal{T})$, if~$\ell<k$ we return $\find_{\mathcal{D}_{k-1}}(\ell,\mathcal{T})$.
  Otherwise, the order of $\CT$ is $k$, and there is a leaf $u$ of
  the tree $T$ such that $\CT$ is the tangle at this leaf. We start
  at the root~$r$ of~$T$ and traverse towards a leaf as
  follows. Suppose we are currently at node~$s$ with children~$t_1$
  and~$t_2$. Exactly one of the two sets~$S(s,t_1)$
  and~$S(s,t_2)=\bar{S(s,t_1)}$ is contained in~$\mathcal{T}$
  and we traverse to the corresponding child. Once we hit a
  leaf,~$u_j$ say, we return~$|\mathcal{D}_{k-1}| +j$.
  \qedhere
\end{enumerate}
\end{proof}

\section{Canonical Tree Decompositions}
\label{sec:tree}

In this section we present the canonical tree decomposition of a
connectivity function into parts corresponding to its tangles of order
at most $k$. Our decomposition is more or less the same as the one
presented by Hundertmark~\cite{hun11}, but our construction differs in two
aspects that are important for the algorithmic treatment. We
exclusively choose leftmost and rightmost minimum separations in our
decomposition. Hundertmark is less restrictive about the separations he
uses, which makes it easier to argue that suitable separations exist,
but infeasible to find them algorithmically.  Moreover, our
construction is modular: we introduce new connectivity functions
during the construction, decompose them and then
merge the decompositions. 

In this section, we often speak of ``canonical'' constructions. The
precise technical meaning depends on the context, but in general a
construction (or algorithm) is \emph{canonical} if every isomorphism between its input objects commutes with an isomorphism between the output objects.
More formally, suppose we have a construction (or algorithm) $A$ that
associates an output $A(I)$ with every input $I$. Then the
construction is \emph{canonical} if for any two inputs $I_1$ and $I_2$ and
every isomorphism $f$ from $I_1$ to $I_2$ there is an isomorphism $g$ from $A(I_1)$ to $A(I_2)$ such that $g(A(I_1))=A(I_2)$, that is, the following diagram
commutes:
\[
\begin{tikzcd}
I_1 \arrow{r}{f} \arrow{d}{A} & I_2 \arrow{d}{A} \\
A(I_1)\arrow{r}{g} & A(I_2)
\end{tikzcd}.
\]
We define an \emph{isomorphism} from a connectivity function $\kappa_1:2^{U_1}\to\NN$
to a connectivity function $\kappa_2:2^{U_2}\to\NN$ to be a 
bijective mapping $f:U_1\to U_2$ such that $\kappa(X)=\kappa(f(X))$
for all $X\subseteq U_1$.
For sets $U_1,U_2$ and families $\CX_1\subseteq 2^{U_1}$, $\CX_2\subseteq
2^{U_2}$ we define an \emph{isomorphism} from $(U_1,\CX_1)$ to $(U_2,\CX_2)$
to be a bijective mapping $f:U_1\to U_2$ such that $X\in\CX_1\iff
f(X)\in\CX_2$ for all $X\subseteq\CX_1$. We assume that the reader is
familiar with tree isomorphisms.

It may be worth noting that our construction of comprehensive tangle
data structures described in the previous section is \emph{not} canonical.

\subsection{Tree Decomposition and Nested Separations}
\label{sec:tree1}

Let $U$ be a finite set. We think of $U$ as being the ground set of a
connectivity function $\kappa$, but this connectivity function plays
no role in this section. A \emph{tree
  decomposition} of $U$ is a pair $(T,\beta)$ consisting of a
tree $T$ and a function $\beta\colon V(T)\to 2^U$ such that the sets
$\beta(t)$ for $t\in V(T)$ are mutually disjoint and their union is
$U$. If $\kappa$ is a connectivity function on $U$, we also call
$(T,\beta)$ a tree decomposition \emph{of $\kappa$}. 

Let $(T,\beta)$ be a tree decomposition of $U$. For every edge $st\in
E(T)$ we let $\tilde\beta(s,t)$ be the union of the sets $\beta(t')$
for all nodes $t'$ in the connected component of $T-st$ that contains
$t$. Note that $\tilde\beta(s,t)=\bar{\tilde\beta(t,s)}$. If
$(T,\beta)$ is a tree decomposition of a connectivity function
$\kappa$ on $U$, we define
the \emph{adhesion} of $(T,\beta)$ to be
$\max\{\kappa(\tilde\beta(s,t))\mid(s,t)\in\vec E(T)\}$. We do not
define a ``width'' for our tree decompositions.

Observe that every branch decomposition $(T,\tilde\xi)$ corresponds to the tree
decomposition $(T,\beta)$ with $\beta(t)=\xi(t)$ for all $t\in
L(T)$ and $\beta(T):=\emptyset$ for all $t\in V(T)\setminus
L(T)$. Therefore, we may just view branch decompositions as special
tree decompositions.  

Our notion of a
tree decomposition of a connectivity function is not new (see for
example~\cite{geegerwhi09}). 
It may be surprising to a reader only
familiar with tree decompositions of graphs, because it partitions the
elements of $U$, whereas the bags of a tree decomposition of a graph
may overlap. But note that if we apply this notion to
the connectivity function $\kappa_G$ of a graph $G$ (see
Example~\ref{exa:1}), we decompose the edge set and not the vertex set
of $G$. The following example details the relation between standard
tree decompositions of graphs and tree decompositions of $\kappa_G$.

\begin{example}
  Let $(T,\beta)$ be a tree decomposition of a graph $G$ (in the usual
  sense). It yields a tree decomposition $(T,\beta')$ of $\kappa_G$ (in
  the sense defined above) as follows: for
  every edge $e\in E(G)$, we arbitrarily choose a node $t_e\in V(T)$
  that covers $e$. Then for every $t\in V(T)$ we let
  $\beta'(t):=\{e\in E(G)\mid t=t_e\}$.

  Conversely, if we have a tree decomposition $(T,\beta')$ of $\kappa_G$,
  then we can define a tree decomposition $(T,\beta)$ of $G$ as
  follows. For every node $v\in V(G)$ we let $U_v$ be the set of all
  nodes $t\in V(T)$ such that $v$ is incident with an edge $e\in
  \beta'(t)$. We let $\hat U_v$ be the union of $U_v$ with all nodes
  $t\in V(T)$ appearing on a path between two nodes in $U_v$. Now we
  let $\beta(t)=\{v\in V(G)\mid t\in\hat U_v\}$. We call $(T,\beta)$
  the \emph{tree decomposition of $G$ corresponding to} $(T,\beta')$.

  Note that the construction of a tree decomposition of $\kappa_G$ from a
  tree decomposition of $G$ involves arbitrary choices, whereas the
  construction of a tree decomposition of $G$ from a
  tree decomposition of $\kappa_G$ is canonical. Thus the ``tree
  decomposition of a graph corresponding to a tree decomposition of
  its edge set'' is well-defined.
  \uend
\end{example}

Let $(T,\beta)$ be a tree decomposition of a set $U$.
We let
\[
\CN(T,\beta)=\{\tilde\beta(s,t)\mid st\in E(T)\}.
\]
and call it the \emph{set of separations} of
$(T,\beta)$. We will now characterise sets of
separations that come from tree decompositions.

Sets (or rather, separations) $X,Y\subseteq U$ are \emph{nested} if
either $X\subseteq Y$ or $X\subseteq\bar Y$ or $\bar X\subseteq Y$ or
$\bar X\subseteq\bar Y$; otherwise $X$ and $Y$ \emph{cross}. Note that
$X$ and $Y$ cross if and only if the four sets $X\cap Y$, $X\cap\bar
Y$, $\bar X\cap Y$, and $\bar X\cap\bar Y$ are all nonempty.  A family
$\CN\subseteq 2^U$ is \emph{nested} if all $X,Y\in\CN$ are nested.
Observe that for every tree decomposition $(T,\beta)$ of $\kappa$ the set
$\CN(T,\beta)$ is nested and closed under complementation.
The following converse of this observation is well-known and goes back
(at least) to \cite{gm10}. (We include a proof for the reader's convenience.)

\begin{lemma}\label{lem:nested}
  If $\CN\subseteq 2^U$ is nested and closed under complementation,
  then there is a tree decomposition
  $(T,\beta)$ of $U$ such that
  $\CN=\CN(T,\beta)$.

  Furthermore, the construction of
  $(T,\beta)$ from $U$ and $\CN$ is
  canonical and can be carried out by a polynomial-time algorithm.
\end{lemma}

\begin{proof}
   By induction on $|\CN|$ we construct a rooted tree $(T,r)$ and a
   mapping $\beta:V(T)\to 2^U$ such that $(T,\beta)$ is a tree
   decomposition with $\CN(T,\beta)=\CN$.
   
   In the base step $\CN=\emptyset$, we let $T$ be a tree with one
   node $r$ and we define $\beta$ by $\beta(r):=U$.

   In the inductive step $\CN\neq\emptyset$, let $X_1,\ldots,X_m$ be a
   list of all inclusion-wise minimal elements of $\CN$ (possibly,
   $m=1$ and $X_1=\emptyset$). As $\CN$ is
   nested, for all $i\neq j$ we have $X_i\subseteq \bar X_j$. This
   implies that the sets $X_i$ are mutually disjoint. %
   Let
   \[
   \CN':=\CN\setminus\{X_i,\bar X_i\mid i\in[m]\}.
   \]
   By the induction hypothesis, there is a rooted tree $(T',r')$ and a
   mapping $\beta':V(T')\to 2^U$ such that $(T',\beta')$ is a tree
   decomposition with $\CN(T',\beta')=\CN'$. For every $i\in[m]$, let
   $t_i$ be a node of minimum
  height (where the \emph{height} of a node is its distance to the root) with $X_i\subseteq\tilde\beta'(s_i,t_i)$ for the parent $s_i$ of $t_i$,
  or $t_i:=r'$ if no such node exists. Observe that there is only one
  such node $t_i$. Indeed, if $t\neq t_i$ has the same height as
  $t_i$, then neither $t_i = r$ nor $t=r$. Let $s$ be the parent of
  $t$. Then the edges $(s_i,t_i)$ and
  $(s,t)$ are pointing away from each other and thus
  $\tilde\beta'(s_i,t_i)\cap \tilde\beta'(s,t)=\emptyset$. This
  implies $X_i\not\subseteq\tilde\beta'(s,t)$, unless $X_i=\emptyset$,
  in which case $t_i$ is the root, and we do not have to worry
  about this.

  We define a new tree $T$ from $T'$ by attaching a fresh leaf $u_i$
  to $t_i$ for every $i\in[m]$. We let $r:=r'$ be the root of $T$. We
  define $\beta\colon V(T)\to2^U$ by
  \[
  \beta(t):=
  \begin{cases}
    X_i&\text{if }t=u_i,\\
    \beta'(t)\setminus\bigcup_{i=1}^mX_i&\text{if }t\in V(T').
  \end{cases}
  \]
  As the sets $X_i$ are mutually disjoint, $(T,\beta)$ is a tree
  decomposition of $U$.  We need to prove that
  $\CN(T,\beta)=\CN$.

  \begin{claim}[1]
    For all oriented edges $(s,t)\in\vec E(T')$ we have
    $\tilde\beta(s,t)=\tilde\beta'(s,t)$.

    \proof
    Let $(s,t)\in\vec E(T')$. Without
  loss of generality we assume that $t$ is a child of $s$. We have
  $\tilde\beta'(s,t)\in\CN'\subseteq\CN$. By the minimality of $X_i$ and the
  nestedness of $\CN$, we have
  $X_i\subseteq \tilde\beta'(s,t)$ or $X_i\cap
  \tilde\beta'(s,t)=\emptyset$. Moreover, $X_i\subseteq \tilde\beta'(s,t)$ if any
  only if $u_i$ is a descendant of $t$ in $T$. As $\tilde\beta'(s,t)$
  is the union of the sets $\beta'(t')$ for all descendants $t'$ of
  $t$ in $T'$ and $\tilde\beta(s,t)$
  is the union of the sets $\beta(t')$ for all descendants $t'$ of
  $t$ in $T$, the claim follows.\uend
  \end{claim}

  To prove that $\CN(T,\beta)\subseteq\CN$, let $Z\in
  \CN(T,\beta)$. Say, $Z=\tilde\beta(s,t)$ for some oriented edge
  $(s,t)\in\tilde E(T)$.  If $(s,t)=(t_i,u_i)$ for some $i\in[m]$,
  then $Z=X_i\in\CN$, and if $(s,t)=(u_i,t_i)$ then $Z=\bar
  X_i\in\CN$, because $\CN$ is closed under
  complementation. Otherwise, $(s,t)\in\vec E(T')$. Then by Claim~1 we
  have $Z=\tilde\beta'(s,t)\in\CN'\subseteq\CN$.

  To prove the converse inclusion, let $Z\in\CN$. If $Z=X_i$ for some
  $i\in[m]$, then $Z=\tilde\beta(t_i,u_i)$, and if $Z=\bar X_i$, then
  $Z=\tilde\beta(u_i,t_i)$. Otherwise, $Z\in\CN'$, and thus by
  Claim~1, $Z=\tilde\beta'(s,t)=\tilde\beta(s,t)$ for some
  $(s,t)\in\vec E(T')$.

  Thus $(T,\beta)$ is a tree decomposition of $U$ with
  $\CN(T,\beta)=\CN$. The construction is obviously canonical, and it
  is easy to see that it can be carried out by a polynomial time algorithm.
\end{proof}

It is our goal to construct tree decompositions whose parts correspond
to tangles and whose separations separate these tangles. If $\KT$ is a
family of mutually incomparable $\kappa$-tangles, then a \emph{tree
  decomposition for $\KT$} is a triple $(T,\beta,\tau)$, where
$(T,\beta)$ is a tree decomposition of
$\kappa$ and $\tau:\KT\to V(T)$ is an injective mapping with the
following properties.
\begin{nlist}{TD}
\item\label{li:td1} For all distinct $\CT,\CT'\in\KT$ there is an oriented edge
  $(t,t')\in\vec E(T)$ on the oriented path from $\tau(\CT)$ to
  $\tau(\CT')$ in $T$ such that $\tilde\beta(t',t)$ is a minimum
  $(\CT,\CT')$-separation.
\item\label{li:td2} For every oriented edge $(t,t')\in\vec E(T)$ there
  are tangles $\CT,\CT'\in\KT$ such that $(t,t')$ appears on the
  oriented path from $\tau(\CT)$ to $\tau(\CT')$ and $\tilde\beta(t',t)$ is a
  minimum $(\CT,\CT')$-separation.
\item\label{li:td3} For every tangle $\CT\in\KT$ and every neighbour
  $t'$ of $t:=\tau(\CT)$ in $T$ it holds that $\tilde\beta(t',t)\in\CT$.
\end{nlist}
Nodes $t\in \tau(\KT)$ are called \emph{tangle nodes} and the
remaining nodes $t\in V(T)\setminus\tau(\KT)$ are called \emph{hub
  nodes}.

For an arbitrary family of $\KT$-tangles, we say that a \emph{tree
  decomposition for $\KT$} is a tree decomposition for the family
$\KT_{\max}\subseteq\KT$ consisting of all inclusion-wise maximal tangles
in $\KT$.

The following example shows that in  general we cannot do without hub
nodes if we want to construct canonical decompositions.

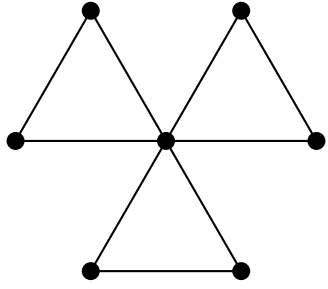
\begin{figure}
  \centering
  \begin{tikzpicture}[
    thick,
    every node/.style={fill=black,circle,minimum width=3pt}
    ]
    \draw[thick] (0,0)--(0:2cm)--(60:2cm)--
          (0,0)--(120:2cm)--(180:2cm)--
          (0,0)--(240:2cm)--(300:2cm)--(0,0);
          
    \draw[fill=black] (0,0) circle (3pt)
                      (0:2cm) circle (3pt) 
                      (60:2cm) circle (3pt) 
                      (120:2cm) circle (3pt) 
                      (180:2cm) circle (3pt) 
                      (240:2cm) circle (3pt) 
                      (300:2cm) circle (3pt) 
    ;
  \end{tikzpicture}
  \caption{The graph of Example~\ref{exa:triangles}}
  \label{fig:triangles}
\end{figure}

\begin{example}\label{exa:triangles}
  Let $G$ be the graph consisting of three triangles joined at a
  single node (see Figure~\ref{fig:triangles}). The connectivity
  function $\kappa_G$ has three tangles of order $2$ corresponding to
  the three triangles. 

  Let $(T,\beta,\tau)$ be a canonical tree decomposition
  for the family of all $\kappa_G$-tangles of order at most $2$. Here
  canonical means that for every automorphism $f$ of the graph $G$
  there is an automorphism $g$ of the tree $T$ such that
  $f(\beta(t))=\beta(g(t))$ for all $t\in V(T)$. Typically, $T$ would
  be a star with three leaves, which are tangle nodes associated with
  the three tangles of order $2$, and one centre, which is a hub node.

  But now suppose for contradiction that the decomposition has no hub
  nodes. Then $T$ must be a path of length $2$. One of the tangles
  must be associated with the centre node of this path, and the other
  two with the leaves. But as the automorphism group of $G$ acts
  transitively on the three tangles, this contradicts the canonicity.  
  \uend
\end{example}

\begin{lemma}\label{lem:td}
  Let $(T,\beta,\tau)$ be a tree decomposition for a family $\KT$ of
  mutually incomparable
  $\kappa$-tangles.
  \begin{enumerate}
  \item For every tangle node $t=\tau(\CT)$ and for every oriented edge
    $(u',u)$ pointing towards $t$, if
    $\kappa(\tilde\beta(u',u))<\ord(\CT)$ then $\tilde\beta(u',u)\in\CT$.
  \item $E(T)=\emptyset$ if and only if $|\KT|\le 1$.
  \item If $\KT\neq\emptyset$, all leaves of\/ $T$ are tangle nodes.
  \item The mapping $\tau:V(T)\to\KT$ is uniquely determined by the
    tree decomposition $(T,\beta)$. That is, if $\tau'$ is an
    injective mappings from $\KT$ to $V(T)$ such that
    $(T,\beta,\tau')$ is a tree decomposition for $\KT$, then
    $\tau=\tau'$.
  \end{enumerate}
\end{lemma}

\begin{proof}
  (1) follows from \ref{li:td3} and the fact that for every edge
  $(u,u')$ pointing towards $t$ there is a neighbour $t'$ of $t$ (the
  neighbour of $t$ on the path from $u'$ to $t$) such that $\tilde\beta(u',u)\supseteq\tilde\beta(t',t)$.
  The forward direction of (2) follows from \ref{li:td1} and the
  backward direction follows from \ref{li:td2}.
  (3) follows from \ref{li:td2}.

  To prove (4), suppose for contradiction that $\tau\neq\tau'$. Let
  $\CT\in\KT$ such that $t:=\tau(\CT)\neq\tau'(\CT)=:t'$. Let $u,u'$
  be the neighbours of $t,t'$, respectively, on the path from $t$ to
  $t'$ in $T$. Then by \ref{li:td3},
  $\tilde\beta(t',t),\tilde\beta(u',u)\in\CT$. However,
  $\tilde\beta(t',t)\cap\tilde\beta(u',u)=\emptyset$. This contradicts
  $\CT$ being a tangle.
\end{proof}

The next lemma shows that a canonical tree decomposition for a family~$\KT$ of
tangles can be constructed from a nested family of separations
satisfying two extra conditions relating it to $\KT$.

\begin{lemma}\label{lem:tn2td}
  Let $\KT$ be a family of $\kappa$-tangles and $\CN\subseteq 2^U$
  a family of separations that is nested and closed under complementation
  and satisfies the following two conditions.
  \begin{nlist}{TN}
    \item\label{li:tn1} For all $\CT,\CT'\in\KT$ with $\CT\bot\CT'$ there is a
      $Z\in\CN$ such that $Z$ is a minimum $(\CT,\CT')$-separation.
    \item\label{li:tn2} For all $Z\in\CN$ there are tangles $\CT,\CT'\in\KT$ with
      $\CT\bot\CT'$ such that $Z$ is a minimum $(\CT,\CT')$-separation.
  \end{nlist}
  Then for every tree decomposition $(T,\beta)$ of $\kappa$ with
  $\CN(T,\beta)=\CN$ there is a unique injective mapping $\tau\colon\KT_{\max}\to
  V(T)$ such that $(T,\beta,\tau)$ is a tree decomposition for $\KT$.

  Furthermore, given $(T,\beta)$ and the index set of $\KT$ in a
  comprehensive tangle data structure, the mapping $\tau$ can be
  computed in polynomial time.
\end{lemma}

\begin{proof}
  Without loss of generality we may assume that the tangles in
  $\KT$ are mutually incomparable; otherwise we work with $\KT_{\max}$
  instead of $\KT$. Observe that $\CN$ satisfies conditions \ref{li:tn1} and
  \ref{li:tn2} with respect to $\KT$ if any only it satisfies the two
  conditions with respect to $\KT_{\max}$.

 We may further  assume that $|\KT|\ge 2$. Then
  $\CN\neq\emptyset$ by \ref{li:tn1}. Furthermore, $\ord(\CT)\ge 1$
  for all $\CT\in\KT$, because the unique tangle of order $0$ is the
  empty tangle, which is comparable with all other tangles. 

  Let $(T,\beta)$ be a tree decomposition of $\kappa$ with
  $\CN(T,\beta)=\CN$. For every $k\ge 1$, we let $E_k$ be the set of
  all edges $e=tt'\in E(T)$ with $\kappa(\tilde\beta(t,t'))<k$.

  For every tangle $\CT\in\KT$ of order $k$ we construct a connected subset
  $X_{\CT}\subseteq V(T)$ as follows: we orient all edges $e=tt'\in
  E_k$ in such a way that they
  point towards $\CT$, that is, if $\tilde\beta(t,t')\in\CT$ then
  the orientation of $e$ is $(t,t')$ and otherwise the orientation is
  $(t',t)$. Then there a unique connected component of $T-E_k$ (the
  forest obtained from $T$ by deleting all edges in $E_k$) such
  that all oriented edges point towards this component. We let
  $X_{\CT}$ be the node set of this connected component.

  It follows from \ref{li:tn1} that the sets $X_{\CT}$ are mutually
  vertex disjoint. To see this, consider distinct
  $\CT,\CT'\in\KT$. Let $Z\in\CN$ be a minimum $(\CT,\CT')$ separation
  and $(t,t')\in\vec E(T)$ such that $\tilde\beta(t',t)=Z$. Then
  $X_{\CT}$ is contained in the connected component of $T-tt'$ that
  contains $t$ and $X_{\CT'}$ is contained in the connected component
  of $T-tt'$ that contains $t'$. Hence $X_{\CT}\cap
  X_{\CT'}=\emptyset$.

  \begin{claim}[1]
    Let $\CT,\CT'\in\KT$ be distinct, and let $Z\in\CN$ be a minimum
    $(\CT,\CT')$-separation. 
    \begin{enumerate}
    \item There is an oriented edge $(t,t')\in
    \vec E(T)$ such that $\tilde\beta(t',t)=Z$.
  \item Every oriented edge $(t,t') \in \vec E(T)$ such that
    $\tilde\beta(t',t)=Z$ appears on the oriented path from $X_{\CT}$
    to $X_{\CT'}$.
    \end{enumerate}

    \proof
    (1) follows immediately from $\CN(T,\beta)=\CN$.

    To prove (2), let $(t,t')\in\vec E(T)$ such that
    $Z=\tilde\beta(t',t)$. 
    As $Z\in\CT$
    the oriented edge $(t',t)$ points towards $X_{\CT}$, and as $\bar
    Z=\tilde\beta(t,t')\in\CT'$ the oriented edge $(t,t')$ points
    towards $X_{\CT'}$. It follows that the oriented edge $(t,t')$
    appears on the oriented path $\vec P$ from $X_{\CT}$ to $X_{\CT'}$
    in $T$.
    \uend
  \end{claim}

  \begin{claim}[2]
    For all $\CT\in\KT$ it holds that $|X_{\CT}|=1$.

    \proof Suppose for contradiction that $|X_{\CT}|>1$ for some
    $\CT\in\KT$. Let $X:=X_{\CT}$. As $X$ is connected, there is an
    edge $e=t_1t_2\in E(T)$ with both endvertices in $X$. Then
    $e\not\in E_{\ord(\CT)}$ and thus
    $\kappa(\tilde\beta(t_1,t_2))\ge\ord(\CT)$.

    Let $\CT_1,\CT_2\in\KT$ such that $Z:=\tilde\beta(t_2,t_1)\in\CN$
    is a minimum $(\CT_1,\CT_2)$-separation. Such tangles exist by
    \ref{li:tn2}. For $i=1,2$, let $X_i:=X_{\CT_i}$. By Claim~1, the
    oriented edge $(t_1,t_2)$ appears on the oriented path $\vec P$
    from $X_1$ to $X_2$ in $T$.

  We have
  \[
  \ord(\CT)\le\kappa(\tilde\beta(t_1,t_2))
  =\kappa(Z)<\min\{\ord(\CT_1),\ord(\CT_2)\}.
  \]
  Let $Z_1\in\CN$ be a minimum $(\CT_1,\CT)$-separation. Then
  $\kappa(Z_1)<\ord(\CT)\le\kappa(Z)$. Moreover, by Claim~1, there is an oriented edge $(u_1,u)$ on the oriented path
  $\vec Q$ from $X_1$ to $X$ such that
  $\tilde\beta(u,u_1)=Z_1$. 

  We have $Z_1\in\CT_1$, because $Z_1$ is a
  $(\CT_1,\CT)$-separation. Since $t_1\in X$, the path $\vec Q$ is an
  initial segment of the path $\vec P$, and therefore $(u_1,u)$ is
  also an edge of $\vec P$. The edge $(u_1,u)$ occurs before $(t_1,t_2)$ on the path
  $\vec P$. Thus $\bar Z_1=\tilde\beta(u_1,u)\supseteq\tilde\beta(t_1,t_2)=\bar
  Z$, and as $\bar Z\in\CT_2$, this implies $\bar Z_1\in\CT_2$. Hence
  $Z_1$ is a $(\CT_1,\CT_2)$-separation. As $\kappa(Z_1)<\kappa(Z)$,
  this contradicts the minimality of $Z$. 
  \uend
  \end{claim}

  We define $\tau:\KT\to V(T)$ by letting $\tau(\CT)$ be the unique
  node in $X_{\CT}$, for all $\CT\in\KT$. This mapping is well-defined
  by Claim~2, and injective, because the sets $X_{\CT}$ are mutually
  disjoint.

  It follows from \ref{li:tn1} and Claim~1 that $(T,\beta,\tau)$ satisfies
  \ref{li:td1}. It follows from \ref{li:tn2} and $\CN(T,\beta)=\CN$
  and Claim~1 that $(T,\beta,\tau)$ satisfies
  \ref{li:td2}.

  By the construction of $X_{\CT}$, for all oriented edges $(t',t)$ with
  $t\in X_{\CT}$ and $t'\not\in X_{\CT}$ it holds that
  $\tilde\beta(t',t)\in\CT$. As $X_{\CT}=\{\tau(\CT)\}$, this implies
  that $(T,\beta,\tau)$ satisfies \ref{li:td3}.

  The uniqueness of $\tau$ follows from Lemma~\ref{lem:td}(4). As
  $\tau(\CT)$ is the unique node $t$ of $T$ such that
  $\tilde\beta(t',t)\in\CT$ for all neighbours $t'$ of $t$, it
  is straightforward to compute $\tau$ in polynomial time.
\end{proof}

We call a family $\CN\subseteq 2^U$ that is nested and closed under
complementation and satisfies \ref{li:tn1} and \ref{li:tn2} a
\emph{nested family for $\KT$}.  Observe the converse of
Lemma~\ref{lem:tn2td}: if $(T,\beta,\tau)$ is a tree decomposition for
$\KT$, then $\CN(T,\beta)$ is a nested family for $\KT$.

\begin{remark}
  It follows from Lemma~\ref{lem:tn2td} that \ref{li:td1} and
  \ref{li:td2} imply \ref{li:td3} and that we can even replace
  \ref{li:td1} and \ref{li:td2} by the weaker conditions \ref{li:tn1}
  and \ref{li:tn2} for $\CN=\CN(T,\beta)$.

  The reason that we nevertheless used \ref{li:td1}--\ref{li:td3} is
  that they state the crucial properties that we expect from a tree
  decomposition for a family of tangles.
\end{remark}

\subsection{Decomposing Coherent Families}
\label{sec:tree2}

Let us call a family $\KT$ of $\kappa$-tangles of order $k+1$
\emph{coherent} if all elements of $\KT$ have the same truncation to order
$k$. Observe that this condition implies, and is in fact equivalent to, the
condition that for distinct $\CT,\CT'\in\KT$ the order of a minimum
$(\CT,\CT')$-separation is $k$.
The main result of this section, Lemma~\ref{lem:coherent}, shows how to
compute  a tree decomposition for a
coherent family
of tangles of order $k+1$. In Section~\ref{sec:tree3}, we will then
combine decompositions for different coherent sets of tangles of different
orders. 

The family of separations of the tree decomposition our algorithm computes for a given
set $\KT$ of tangles will be a subset of the set 
\[
\CZ(\KT)=\{ Z(\CT,\CT')\mid \CT,\CT'\in\KT\text{
such that }\CT\bot\CT'\}.
\]
of all leftmost minimum separations 
of pairs of tangles in $\KT$ and of their complements.
The following lemma is similar to Lemma~5.3 of \cite{hun11}. But our
proof is different, because we work with
different assumptions and a different set of separations.

\begin{lemma}\label{lem:coherent0}
  Let $\KT$ be a coherent family of $\kappa$-tangles of order $k+1$, and
  let $Z_0\in\CZ(\KT)$ be inclusion-wise minimal. Then for all
  $Z\in\CZ(\KT)$, either $Z_0\subseteq Z$ or $Z_0\subseteq \bar Z$.
\end{lemma}

\begin{proof}
  Let $\CT_0,\CT_0'\in\KT$ such
  that $Z_0=Z(\CT_0,\CT_0')$. Moreover, let $\CT,\CT'\in\KT$ be
  distinct, and let $Z$ be a minimum $(\CT,\CT')$-separation. We shall
  prove that $Z_0\subseteq Z$ or
  $Z_0\subseteq \bar Z$. Of course this will imply the assertion of
  the lemma, because every $Z\in\CZ(\KT)$ is a minimum
  $(\CT,\CT')$-separation for some $\CT,\CT'\in\KT$.
  
  Without loss of generality, we may assume that
  \begin{equation}
    \label{eq:coh1}
    Z\in\CT_0.
  \end{equation}
  Otherwise, we swap $\CT$ and $\CT'$ and take $\bar Z$ instead of
  $Z$.
  
  Suppose first that $\kappa(Z_0\cap Z)\le k$. Then $Z_0\cap
  Z\in\CT_0$, because $\bar{Z_0\cap Z}\cap Z_0\cap
  Z=\emptyset$. Moreover, we have $\bar{Z_0\cap Z}\in\CT_0'$,
  because $(Z_0\cap Z)\cap\bar Z_0=\emptyset$.  Thus $Z_0\cap Z$
  is a $(\CT_0,\CT_0')$-separation. As $Z_0$ is leftmost minimum, it
  follows that $Z_0\subseteq Z_0\cap Z$ and thus $Z_0\subseteq Z$.

  In the following, we assume $\kappa(Z_0\cap Z)> k=\kappa(Z)$. By
  submodularity, $ \kappa(Z_0\cup Z)<\kappa(Z_0) $. We have $Z_0\cup
  Z\in\CT_0$, because $\bar{Z_0\cup Z}\cap Z_0=\emptyset$. If $Z_0\cup
  Z$ was a $(\CT_0,\CT_0')$-separation, $\kappa(Z_0\cup Z)<\kappa(Z_0)$
  would contradict the minimality of $Z_0$. Hence $\bar{Z_0\cup
    Z}\not\in\CT_0'$, which implies $Z_0\cup Z\in\CT_0'$. As $\bar
  Z_0\in\CT_0'$ and $\bar Z_0\cap \bar Z\cap (Z_0\cup Z)=\emptyset$,
  it follows that
    \begin{equation}
      \label{eq:coh2}
      Z\in\CT_0'.
    \end{equation}
     \begin{cs}
       \case1
       $Z_0\in\CT'$.\\
       Then $Z_0$ is a $(\CT',\CT_0')$-separation. As $\CT',\CT_0'$
       have the same truncation to order $k$, $Z_0$ is in fact a
       minimum $(\CT',\CT_0')$-separation, and this implies
       $Z(\CT',\CT_0')\subseteq Z_0$. By the inclusion-wise minimality
       of $Z_0$ in $\CZ(\KT)$, this implies $Z_0=Z(\CT',\CT_0')$

       By \eqref{eq:coh2}, $\bar Z$ is another
       $(\CT',\CT_0')$-separation. As $\kappa(Z)=\kappa(Z_0)$, it
       follows that $Z_0=Z(\CT',\CT_0')\subseteq \bar Z$.
      \case2
      $\bar Z_0\in\CT'$.\\
       Then $Z_0$ is a $(\CT_0,\CT')$-separation, and by a similar
       argument as in Case~1 it follows that $Z_0=Z(\CT_0,\CT')$. 

       By \eqref{eq:coh1}, $Z$ is another
       $(\CT_0,\CT')$-separation, and it
       follows that $Z_0=Z(\CT_0,\CT')\subseteq Z$.
       \qedhere
    \end{cs}
\end{proof}

\begin{lemma}\label{lem:coherent}
  Let $k\ge0$.
  There is a polynomial time algorithm that, given 
  a coherent family $\KT$ of $\kappa$-tangles of order $k+1$
  for a connectivity function $\kappa$ on a set $U$ (via a
  comprehensive tangle data structure and the set of indices of the
  tangles in $\KT$),
  computes a canonical nested family for $\KT$.
\end{lemma}

\begin{proof}
  The idea of the proof is to construct a tree decomposition for $\KT$
  starting from the leaves of the decomposition tree and then moving
  towards the centre of the tree. Observe that the separations of a tree
  decomposition associated with the edges towards the leafs are
  precisely the inclusion-wise minimal separations.

  The algorithm inductively computes for all $i\in\NN$ a set
  $\CN_i\subseteq\CZ(\KT)$ of separations and a family $\KT_i$ of
  tangles.
  \begin{itemize}
  \item $\CN_0:=\emptyset$ and $\KT_0:=\emptyset$.
  \item Suppose that $\CN_i$ and $\KT_i$ are already computed. Then
    the algorithm repeatedly queries the tangle data structure to
    obtain $\CZ(\KT\setminus\KT_i)$. 

    $\CN_{i+1}$ is the union of $\CN_i$ with all
    inclusion-wise minimal $Z\in\CZ(\KT\setminus\KT_i)$, and
    $\KT_{i+1}$ is the set of all tangles $\CT\in\KT$ such that
    $Z(\CT,\CT')\in\CN_{i+1}$ for
    some $\CT'\in\KT$.
  \end{itemize}
  Let $\CN$ be the closure of  $\bigcup_{i\ge0}\CN_i$ under
  complementation. It is easy to see that $\CN$ can be computed in
  polynomial time. We claim that $\CN$ is a nested family for $\KT$.

  It follows from Lemma~\ref{lem:coherent0} that $\CN$ is nested:
  when we add  a $Z_0$ to $\CN_{i+1}$, it is nested with all
  $\CZ(\KT\setminus\KT_i)$ and thus with all $Z\in\bigcup_{j\ge
    i+1}\CN_j$.

  The family $\CN$ trivially satisfies \ref{li:tn2}, because each element of each $\CN_i$ is an element of $\CZ(\KT)$.

  It remains to prove that $\CN$ satisfies \ref{li:tn1}. For all $i\ge0$
  we prove that for all $\CT\in\KT_{i+1}\setminus\KT_i$,
  $\CT'\in\KT\setminus\KT_i$ there is a $Z\in\CN_{i+1}$ such
  that $Z$ 
  is a minimum $(\CT,\CT')$-separation. Let $\CT\in\KT_{i+1}\setminus\KT_i$,
  $\CT'\in\KT\setminus\KT_i$. Let $\CT''\in\KT$
  such that $Z=Z(\CT,\CT'')\in\CN_{i+1}$. Then $Z$ is inclusion-wise
  minimal in $\CZ(\KT\setminus\KT_i)$. Let
  $Z'=Z(\CT,\CT')$. By Lemma~\ref{lem:coherent0}, either $Z\subseteq
  Z'$ or $Z\subseteq \bar Z'$. If $Z\subseteq Z'$, then $\bar
  Z\in\CT'$, because $\bar Z\supseteq\bar Z'\in\CT'$, and thus $Z$ is
  a $(\CT,\CT')$-separation. If $Z\subseteq \bar Z'$, then $Z\cap
  Z'=\emptyset$, which contradicts $Z,Z'\in\CT$.

  Now we observe that $\Big|\KT\setminus\bigcup_{i\ge0}\KT_i\Big|\le
  1$, because otherwise $\CZ(\KT\setminus\bigcup_{i\ge0}\KT_i)$ would
  be nonempty and the constructions would not have stopped. So for distinct tangles $\CT,\CT'\in\KT$, at least one
  of them is $\bigcup_{i\ge0}\KT_i$, and for some $i\ge0$ either
  $\CT\in\KT_{i+1}\setminus\KT_i$ and
  $\CT'\in\KT\setminus\KT_i$ or vice versa. We have just seen that
  then $\CN_{i+1}$ contains a $(\CT,\CT')$-separation or a
  $(\CT',\CT)$-separation, and this implies that $\CN$, which is
  closed under complementation, contains a $(\CT,\CT')$-separation. 
\end{proof}

\subsection{Decomposing Arbitrary Families}
\label{sec:tree3}

\begin{figure}
  \centering
  \begin{tikzpicture}[thick]
  \begin{scope}[scale = 1.3]
  \draw[blue]  (0,-0.25) node  {} ellipse (1.75 and 1.25);
  \node[blue] at (-0.1,0.75) {$B$};
  \draw[fill,inner sep = -1pt]  (0,-0.25) node (t) {} circle (0.05 and .05);
  \node at (0.3,-0.15) {$t$};

  \node[blue] at (-1.25,2.25) {$\overline{C_5}$};
  \draw[blue]  plot[smooth, tension=.7] coordinates {(-0.75,2.5) (-0.75,0.25) (-2.75,1.25)};
  \draw[fill,inner sep = -1pt]  (-1.5,1.25) node (t5) {} circle (0.05 and .05);

  \node at (-1.25,1.4) {$t_5$};
  \draw[node distance = 0cm, inner sep = -1pt]  (t) -- (t5);
  \draw[node distance = 0cm, inner sep = -1pt]  (t5) -- (-2,1.25);
  \draw[node distance = 0cm, inner sep = -1pt]  (t5) -- (-1.5,1.75);
  \draw[node distance = 0cm, inner sep = -1pt]  (t5) -- (-1.8,1.6);

  \node[blue] at (-2.5,-0.25) {$\overline{C_4}$};
  \draw[blue]  plot[smooth, tension=.7] coordinates {(-3,0.25) (-1.25,-0.5) (-2.5,-1.5)};
  \draw[fill,inner sep = -1pt]  (-2.5,-0.75) node (t4) {} circle (0.05 and .05);
  \node at (-2.3,-1) {$t_4$};
  \draw[node distance = 0cm, inner sep = -1pt]  (t) -- (t4);
  \draw (-3,-0.5) -- (t4) -- (-2.75,-1.25);
  
  \node[blue] at (-0.65,-1.85) {$\overline{C_3}$};
  \draw[blue]  plot[smooth, tension=.7] coordinates {(-1.25,-2.75) (-0.75,-1.25) (-0.15,-0.75) (0.5,-1.25) (0.5,-3)};
  \draw[fill,inner sep = -1pt]  (-0.25,-2) node (t3) {} circle (0.05 and .05);
  \node at (0,-1.85) {$t_3$};
  \draw[node distance = 0cm, inner sep = -1pt]  (t) -- (t3);
  \draw (-0.75,-2.4) -- (t3) -- (-0.3,-2.4) -- (t3) -- (0.15,-2.4);

  \node[blue] at (2.25,-0.5) {$\overline{C_2}$};
  \draw[blue]  plot[smooth, tension=.7] coordinates {(2.25,-2.25) (1.25,-1.5) (1,-0.75) (1.5,-0.25) (2.5,-0.25) (3.25,-0.5)};
  \draw[fill,inner sep = -1pt]  (2.5,-1.25) node (t2) {} circle (0.05 and .05);
  \node at (2.25,-1.45) {$t_2$};
  \draw[node distance = 0cm, inner sep = -1pt]  (t) -- (t2);
  \draw (2.75,-1.75) -- (t2) -- (3,-1.5) -- (t2) -- (3,-1);
  
  \node[blue] at (1.5,1) {$\overline{C_1}$};
  \draw[blue]  plot[smooth, tension=.7] coordinates {(0.5,2.25) (0.25,1.5) (0.35,0.75) (1,0.35) (1.75,0.75) (2.5,1.75)};
  \draw[fill,inner sep = -1pt]  (1.25,1.5) node (t1) {} circle (0.05 and .05);
  \node at (1,1.5) {$t_1$};
  \draw[node distance = 0cm, inner sep = -1pt]  (t) -- (t1);
  \draw (1,2.25) -- (t1) -- (1.5,2) -- (t1) -- (2,1.5);
  \end{scope}
    \path (0,-5) node[anchor=north] {(a) \parbox[t]{5cm}{Node $t$ and its neighbours
      $t_i$ and the sets $B,C_i$}};

    \begin{scope}[scale = 1, shift = {(8.5,0)}]
    
    \draw[blue]  (0,-0.25) node  {} ellipse (1.75 and 1.25);
    \node[blue] at (0,0.75) {$B$};
    \draw[blue,fill,inner sep = -1pt]  (-1.5,1.25) node (t5) {} circle (0.05 and .05);
    \node at (-1.25,1.4) {$c_5$};
    \draw[blue,fill,inner sep = -1pt]  (-2.5,-0.75) node (t4) {} circle (0.05 and .05);
    \node at (-2.3,-1) {$c_4$};
    \draw[blue,fill,inner sep = -1pt]  (-0.25,-2) node (t3) {} circle (0.05 and .05);
    \node at (0,-1.85) {$c_3$};
    \draw[blue,fill,inner sep = -1pt]  (2.5,-1.25) node (t2) {} circle (0.05 and .05);
    \node at (2.25,-1.45) {$c_2$};
    \draw[blue,fill,inner sep = -1pt]  (1.25,1.5) node (t1) {} circle (0.05 and .05);
    \node at (1,1.5) {$c_1$};
    \end{scope}
    \path (8.5,-5) node[anchor=north] {(b) Contraction of the $C_i$};
    \draw[dashed] (5,-4) -- (5,3.5);
  \end{tikzpicture}
  \caption{Sets at a node $t$ of a tree decomposition}
  \label{fig:treedec}
\end{figure}
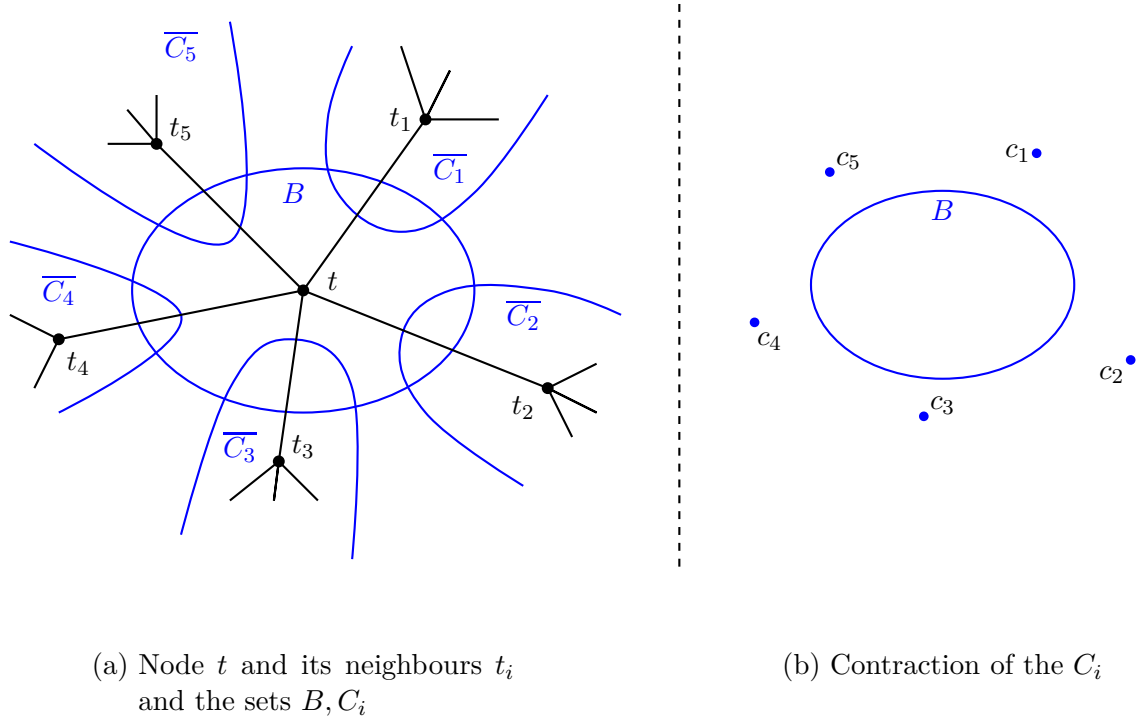

In this section, we will describe how to build a ``global'' tree
decomposition of all tangles of order at most $k+1$ from ``local''
decompositions for coherent families of tangles. Suppose that we have
already built a tree decomposition $(T^{\le k},\beta^{\le k},\tau^{\le
  k})$ for the family $\KT^{\le k}$ of
all $\kappa$-tangles of order at most $k$. Consider a tangle node $t$ of this decomposition and let
$\CT_t\in\KT^{\le k}_{\max}$ such that $\tau^{\le k}(\CT_t)=t$. Suppose that $\ord(\CT_t)=k$, and let
$\KT_t$ be the set of all $\kappa$-tangles of order $k+1$ whose
truncation to order $k$ is $\CT_t$. We
want to expand our decomposition to a decomposition over the set
$\KT^{\le k}\cup\KT_t$. In fact, we want to do this simultaneously for all
nodes $t$ of $T$ in a consistent way to obtain a tree decomposition
for $\KT^{\le k+1}$, but let us focus on just one
node $t$ first. Suppose that the neighbours of $t$ in $T^{\le k}$ are
$t_1,\ldots,t_m$. For every $i\in[m]$, let $C_i=\tilde\beta(t_i,t)$,
and let $B=\beta(t)$ (see Figure~\ref{fig:treedec}(a)). Then $B=\bigcap_{i=1}^mC_i$ and
$U=B\cup\bigcup_{i=1}^m\bar C_i$. Furthermore, the sets $\bar C_i$ for
$i\in[m]$ are mutually disjoint. 

We now ``contract'' each of the sets $\bar C_i$ to a single vertex and define a new
connectivity function on the resulting set.  We take fresh points
$c_1,\ldots,c_m$ not in $U$ and let
$U\contract_t:=B\cup\{c_1,\ldots,c_m\}$ (see Figure~\ref{fig:treedec}(b)). We call $U\contract_t$ the
\emph{contraction of $U$ at $t$}. We define the \emph{expansion} of a
set $X\subseteq U\contract_t$ to be the set
$
X\expand_t:=(X\cap B)\cup\bigcup_{c_i\in X}\bar C_i.
$
We define a set function $\kappa\contract_t$ on $U\contract_t$ by
letting $\kappa\contract_t(X):=\kappa(X\expand_t)$. It is easy to
verify that $\kappa\contract_t$ is a connectivity function on $U\contract_t$.
For every $\kappa$-tangle $\CT$ we let
\[
\CT\contract_t:=\{X\subseteq U\contract_t\mid X\expand_t\in\CT\}.
\]
$\CT\contract_t$ is not necessarily a $\kappa\contract_t$-tangle,
because it may violate tangle axiom \ref{li:t3}. However, it is easy
to see that $\CT\contract_t$ is a 
$\kappa\contract_t$-tangle of the same order as $\CT$ if $\bar C_i\not\in\CT$ for all
$i\in[m]$. By \ref{li:td3}, for all $\CT\in\KT_t$ we have
$C_i\in\CT_t\subseteq\CT$ and thus $\bar C_i\not\in\CT$  for all $i$. Thus $\contract_t$ defines a
mapping from $\KT_t$ to the set of all $\kappa\contract_t$-tangles of
order $k+1$.
Lemma~\ref{lem:injective} below implies that this ``contraction mapping'' is injective.

\begin{lemma}\label{lem:intersection}
  Let $X\subseteq U$ such that there are $\kappa$-tangles $\CT,\CT'$
  for which $X$ is a minimum $(\CT,\CT')$-separation. Then for every
  $Y\subseteq U$, either $\kappa(X\cap Y)\le\kappa(Y)$ or
  $\kappa(X\setminus Y)\le\kappa(Y)$.
\end{lemma}

\begin{proof}
  Suppose for contradiction that $\kappa(X\cap Y)>\kappa(Y)$ and
  $\kappa(X\setminus Y)>\kappa(Y)$. Then by submodularity,
  $\kappa(X\cup Y)<\kappa(X)$ and $\kappa(X\cup\bar Y)<\kappa(X)$. 

  Now let $\CT,\CT'$ be tangles $\CT,\CT'$
  such that $X$ is a minimum $(\CT,\CT')$-separation.
  As $X\subseteq X\cup Y,X\cup\bar Y$ and $X\in\CT$, we have $X\cup
  Y,X\cup\bar Y\in\CT$. 
  Furthermore, either $\bar{X\cup Y}\in\CT'$ or $\bar{X\cup\bar
    Y}\in\CT'$, because 
  $
  \bar X\cap(X\cup Y)\cap(X\cup\bar Y)=\emptyset.
  $
  Thus either $X\cup Y$ or $X\cup\bar Y$ is a $(\CT,\CT')$-separation
  of order less than $\kappa(X)$. This contradicts the minimality
  of $X$.
\end{proof}

\begin{lemma}\label{lem:injective}
  Let $\CT,\CT'\supseteq\CT_t$ be incomparable. Then $\CT\contract_t$ and
  $\CT'\contract_t$ are incomparable as well, and for every 
  minimum $(\CT\contract_t,\CT'\contract_t)$-separation $Z$ the expansion
  $Z\expand_t$ is a minimum $(\CT,\CT')$-separation.
\end{lemma}

\begin{proof}
  We choose a minimum $(\CT,\CT')$-separation $Y$ in such a way that
  it maximises the number of $i\in[m]$ with $Y\cap \bar C_i=\emptyset$
  or $\bar C_i\subseteq Y$. 

  \begin{claim}[1]
    For all $i\in[m]$, either $Y\cap\bar C_i=\emptyset$ or $\bar
    C_i\subseteq Y$.

    \proof Suppose for contradiction that there is some $i\in[m]$ such
    that $\emptyset\subset \bar C_i\cap Y\subset \bar C_i$. By
    \ref{li:td2}, there are tangles $\CT_i,\CT_i'$ such that
    $C_i=\tilde\beta^{\le k}(t_i,t)$ is a minimum
    $(\CT_i,\CT_i')$-separation.  By Lemma~\ref{lem:intersection}
    (applied to $X:=C_i$ and $Y$), either $\kappa(Y\cap
    C_i)\le\kappa(Y)$ or $\kappa(\bar Y\cap C_i)\le\kappa(Y)$. 
    
    Suppose first that $\kappa(Y\cap
    C_i)\le\kappa(Y)$. Then $Y\cap
    C_i\in\CT$, because $Y\in\CT$ and $C_i\in\CT_t\subseteq\CT$ and
    $Y\cap C_i\cap\bar{Y\cap C_i}=\emptyset$. Furthermore, $\bar{Y\cap
      C_i}\in\CT'$, because $\bar Y\cap (Y\cap C_i)=\emptyset$. Thus
    $(Y\cap C_i)$ is a minimum $(\CT,\CT')$-separation as
    well. Furthermore, $(Y\cap C_i)\cap\bar C_i=\emptyset$, and for
    all $j\neq i$, if $Y\cap\bar C_j=\emptyset$ then $(Y\cap
    C_i)\cap\bar C_j=\emptyset$, and if $\bar C_j\subseteq Y$, then
    $\bar C_j\subseteq (Y\cup C_i)$, because $\bar
    C_j=\tilde\beta(t,t_j)\subseteq\tilde\beta(t_i,t)=C_i$. This
    contradicts the choice of $Y$.

    Suppose next that $\kappa(\bar Y\cap
    C_i)\le\kappa(Y)$. Arguing as above with $Y,\bar Y$ and $\CT,\CT'$
    swapped, we see that $\bar Y\cap C_i$ is a minimum
    $(\CT',\CT)$-separation. Thus $Y\cup\bar C_i$ is a minimum
    $(\CT,\CT')$-separation. We have $\bar C_i\subseteq Y\cup\bar
    C_i$, and for all $j\neq i$, if $\bar C_j\subseteq Y$ then $\bar
    C_j\subseteq \bar C_i\cup Y$, and if $\bar C_j\cap Y=\emptyset$
    then $\bar C_j\cap(\bar C_i\cup Y)=\bar C_j\cap \bar
    C_i=\emptyset$. Again, this contradicts the choice of $Y$.
    \uend
  \end{claim}

  It follows from Claim~1 that there is a $Y'\subseteq U\contract_t$
  such that $Y=Y'\expand_t$.
   This set $Y'$ is a
  $(\CT\contract_t,\CT'\contract_t)$-separation. Thus
  $\CT\contract_t\bot\CT'\contract_t$, and the order $k'$ of a minimum
  $(\CT\contract_t,\CT'\contract_t)$-separation is at most
  $\kappa(Y')=\kappa(Y)=:k$, the order of a minimum
  $(\CT,\CT')$-separation.
  Now let $Z\subseteq U\contract_t$ be a minimum
  $(\CT\contract_t,\CT'\contract_t)$-separation. Then the expansion
  $Z\expand_t$ is a
  $(\CT,\CT')$-separation. Thus
  \[
  k\le\kappa(Z\expand_t)=\kappa\contract_t(Z)=k'\le k.
  \]
 Hence $k=k'$, and $Z\expand_t$ is a
  minimum $(\CT,\CT')$-separation.
\end{proof}

\begin{corollary}\label{cor:injective}
   The ``contraction mapping'' $\contract_t$ on $\KT_t$ is
  injective.
\end{corollary}

\begin{proof}
  Let $\CT,\CT'\in\KT_t$
  with $\CT\neq\CT'$. Then $\CT\bot\CT'$, because
  $\ord(\CT)=\ord(\CT')=k+1$, and it follows from
  the lemma that $\CT\contract_t\neq\CT'\contract_t$.
\end{proof}

Finally, we are ready to prove the main theorem this section.

\begin{theorem}\label{theo:candec}
  Let $\ell\ge 0$. Then there is a polynomial time algorithm that,
  given oracle access to a connectivity function $\kappa$,
  computes a canonical tree decomposition for the set of all
  $\kappa$-tangles of order at most $\ell$.
\end{theorem}

\begin{proof}
  Our algorithm first computes a comprehensive tangle data structure
  of order $\ell$. Then, by induction on $0\le k\le \ell$, it computes a
  tree decomposition $(T^{\le k},\beta^{\le^k},\tau^{\le
    k})$
  for the set $\KT^{\le
    k}$ of all $\kappa$-tangles of order at most $k$.

  The base step $k=0$ is trivial, because the only tangle of order $0$
  is the empty set, and the trivial one-node tree decomposition is a
  tree decomposition for this tangle.

  For the inductive step $k\to k+1$ (where $k<\ell$), we assume that
  we have already constructed a canonical tree decomposition $(T^{\le
    k},\beta^{\le^k},\tau^{\le k})$ for $\KT^{\le k}$. Let $\CN^{\le k}:=\CN(T^{\le
    k},\beta^{\le k})$ be the set of separations of this tree
  decomposition.

  For every tangle node $t\in \tau^{\le k}(\KT^{\le k}_{\max})$ we do the following:
  \begin{itemize}
  \item We compute the set $\KT_t$ of all tangles
    $\CT\supseteq\CT_t:=(\tau^{\le k})^{-1}(t)$ of order
    $k+1$.
  \item We compute $U\contract_t$; we can use our oracle for $\kappa$
    to implement an oracle for $\kappa\contract_t$. Then we compute a
    comprehensive tangle data structure of order $k+1$ for $\kappa\contract_t$.
  \item Within this data structure, we compute the family $\KT^\vee_t$ of all
    tangles $\CT\contract_t$ for $\CT\in\KT_t$. 

    Note that $\KT_t^\vee$ is a coherent family of
    $\kappa\contract_t$-tangles of order $k+1$.
  \item We apply Lemma~\ref{lem:coherent} to $\kappa\contract_t$ and
    $\KT_t^\vee$ and obtain a nested family $\CN_t^\vee$ for
    $\KT_t^\vee$. 
  \item We compute the set $\CN_t:=\{X\expand_t\mid X\in\CN_t^\vee\}$.
  \end{itemize}

  \begin{claim}[1]
    The family
    $
    \CN^{\le k+1}:=\CN^{\le k}\cup\bigcup_{t\in V(T^{\le t})}\CN_t
    $
    of separations is nested.

    \proof
    We already know that the family $\CN^{\le k}$ is nested. Furthermore,
    for every
    $t\in V(T^{\le k})$ the family $\CN_t^\vee$ is nested, and this
    implies that $\CN_t$ is nested as well.

    Thus we need to show that the sets in $\CN_t$ are nested with all
    sets in $\CN^{\le k}$ as well as all sets in $\CN_u$ for $u\neq t$. So let
    $X\in\CN_t$. As before, let $t_1,\ldots,t_m$ be the neighbours of
    $t$ in $T^{\le k}$. Then for
    all $i\in[m]$, either $\tilde\beta(t,t_i)\subseteq X$ or $X\cap \tilde\beta(t,t_i)=\emptyset$ and thus $\tilde\beta(t,t_i)\subseteq X$ or
    $\tilde\beta(t,t_i)\subseteq\bar X$. 

    Every set $Z\in\CN^{\le k}$ is
    of the form $\tilde\beta(s,s')$ for some $(s,s')\in\vec E(T^{\le
      k})$. Then there is an $i\in[m]$ such that either
    $Z=\tilde\beta(s,s')\subseteq\tilde\beta(t,t_i)$ or $\bar
    Z=\tilde\beta(s',s)\subseteq\tilde\beta(t,t_i)$. As
    $\tilde\beta(t,t_i)\subseteq X$ or
    $\tilde\beta(t,t_i)\subseteq\bar X$, it follows that $Z$ and $X$
    are nested.

    Now consider a set $Y\in\CN_u$ for some node $u\neq t$. Let $t_i$
    be the neighbour of $t$ and $u_j$ the neighbour of $u$ on the path
    from $t$ to $u$ in $T^{\le k}$. Then either $\tilde\beta(u,u_j)\subseteq Y$ or
    $\tilde\beta(u,u_j)\subseteq\bar Y$, which implies $\bar
    Y\subseteq
    \bar{\tilde\beta(u,u_j)}=\tilde\beta(u_j,u)\subseteq\tilde\beta(t,t_i)$
    or $Y\subseteq\tilde\beta(u_j,u)\subseteq\tilde\beta(t,t_i)$. As
    $\tilde\beta(t,t_i)\subseteq X$ or
    $\tilde\beta(t,t_i)\subseteq\bar X$, it follows that $Y$ and $X$
    are nested.
    \uend
  \end{claim}

  \begin{claim}[2]
   $\CN^{\le k+1}$ is a nested family for $\KT^{\le k+1}$.

    \proof
    By Claim~1, the family $\CN^{\le k+1}$ is nested. By construction,
    it is closed under complementation (instead of going through the
    construction to check this, we can also just close it under
    complementation without any harm). We have to prove that it
    satisfies \ref{li:tn1} and \ref{li:tn2} for $\KT:=\KT^{\le k+1}$.

    It satisfies \ref{li:tn2}, because $\CN^{\le k}$ does and by
    Lemma~\ref{lem:injective}, for all $t\in V(T^{\le k})$, all
    $Z\in\CN_t$ are minimum separations for tangles in
    $\KT_t\subseteq\KT^{\le k+1}$.
   
    To see that $\CN^{\le k+1}$ satisfies \ref{li:tn1}, let $\CT,\CT'\in\KT^{\le
      k+1}$ be incomparable. Let $\CT_0$ be the truncation of $\CT$ to order $k$ if
    $\ord(\CT)=k+1$ and $\CT_0:=\CT$ otherwise, and let $\CT_0'$ be
    defined similarly from $\CT'$. If $\CT_0\bot\CT_0'$, there is a
    $Z\in\CN^{\le k}$ that is a minimum $(\CT_0,\CT_0')$-separation,
    and this $Z$ is also a minimum $(\CT,\CT')$-separation. Otherwise,
    $\ord(\CT)=\ord(\CT')=k+1$ and $\CT_0=\CT_0'$. Let $t:=\tau^{\le
      k}(\CT_0)$. Then $\CT,\CT'\in\KT_t$, and $\CN_t$ contains a
    minimum $(\CT,\CT')$-separation.
    \uend
\end{claim}

    We use the algorithm of Lemma~\ref{lem:nested} to a compute a
    canonical tree decomposition $(T^{\le k+1},\beta^{\le k+1})$ of
    $\kappa$ with $\CN(T^{\le k+1},\beta^{\le k+1})=\CN^{\le k+1}$,
    and we use the algorithm of Lemma~\ref{lem:tn2td} to turn it into
    a tree decomposition $(T^{\le k+1},\beta^{\le k+1},\tau^{\le
      k+1})$ for $\KT^{\le k+1}$.
 \end{proof}

Let us say that a
$\kappa$-tangle is \emph{$\ell$-maximal} if it is
inclusion-wise maximal among all $\kappa$-tangles of order at most
$\ell$. In our previous notation, $\KT^{\le\ell}_{\max}$ denotes the set of all
$\ell$-maximal $\kappa$-tangles. 

 \begin{remark}\label{rem:linear}
   The existence of a canonical tree decomposition for the set
   $\KT^{\le\ell}$ of all
   tangles of order at most $\ell$, which of course follows from
   Theorem~\ref{theo:candec}, implies that for every $\ell$ the
   number $|\KT^{\le\ell}_{\max}|$ of $\ell$-maximal $\kappa$-tangles is at most $|U|-1$,
   provided $|U|\ge 2$. 

   To see this, we assume without loss of generality that
   $|\KT^{\le\ell}_{\max}|\ge 2$. Let $(T,\beta,\tau)$ be a tree
   decomposition for $\KT^{\le\ell}$. Then $|T|\ge 2$. By
   Lemma~\ref{lem:td}, all leaves of $T$ are tangle nodes. Let $t$ be
   a leaf, $s$ the neighbour of $t$, and let $\CT_t\in
   \KT^{\le\ell}_{\max}$ be the tangle with $\tau(\CT_t)=t$. By
   \ref{li:td3} we have $\beta(t)=\tilde\beta(s,t)\in\CT_t$. By
   \ref{li:t2} and \ref{li:t3}, this implies $|\beta(t)|>1$.

   Now let $t\in V(T)$ be a tangle node of degree $2$, say, with
   neighbours $s$ and $u$. Again, let $\CT_t\in \KT^{\le\ell}_{\max}$
   be the tangle $\tau(\CT_t)=t$.  By \ref{li:td3}, we have
   $\tilde\beta(u,t),\tilde\beta(s,t)\in\CT_t$, which implies
   $\beta(t)=\tilde\beta(u,t)\cap\tilde\beta(s,t)\neq\emptyset$.

   Let $n_1,n_2,n_{\ge 3}$ be the numbers of tangle nodes of degree
   $1$, $2$, at least $3$, respectively. We have $2n_1+n_2\le |U|$.
   Furthermore, $n_{\ge 3}<n_1$, because a tree with $n_1$ leaves has
   less than $n_1$ vertices of degree at least $3$. Thus
   \[
   |\KT^{\le\ell}_{\max}|=n_1+n_2+n_{\ge 3}<2n_1+n_2\le|U|.\uende
   \]
 \end{remark}

We close this section with another decomposition algorithm
that may be useful in some applications. Theorem~\ref{theo:candec} yields a tree decomposition
$(T,\beta)$ where at
most one $\ell$-maximal $\kappa$-tangle is associated with every
node. However, in applications we may want to work with the ``local
structure'' at the nodes $t$ of the decomposition, and this local
structure is represented by the ``contractions'' $\kappa\contract_t$
on $U\contract_t$. To understand this local structure, we might be
more interested in $\kappa\contract_t$-tangles than in $\kappa$-tangles
associated with $t$. It is not clear whether at every tangle node $t$
there  is at most one $\ell$-maximal $\kappa\contract_t$-tangle, and we
know even less about the hub nodes of the decomposition. However, the
following theorem shows that we can also construct a decomposition
where at every node $t$ we have at most one $\ell$-maximal
$\kappa\contract_t$-tangle. 

\begin{theorem}\label{theo:candec2}
    Let $\ell\ge 0$. Then there is a polynomial time algorithm that,
    given oracle access to a connectivity function $\kappa$,
  computes a canonical tree decomposition $(T,\beta)$ of $\kappa$ of
  adhesion less than $\ell$ such
  that for all $t\in V(T)$ there is
  exactly one $\ell$-maximal $\kappa\contract_t$-tangle.
\end{theorem}

\begin{proof}
  We start by computing the decomposition of $(T^0,\beta^0,\tau^0)$ of
  Theorem~\ref{theo:candec}. Let $\CN^0:=\CN(T^0,\beta^0)$. It follows
  from \ref{li:td2} that the adhesion of this decomposition is less
  than $\ell$. As every
  connectivity function has at least one tangle (the empty tangle of
  order $0$), for every node $t\in V(T)$ there is at least one
  $\ell$-maximal $\kappa\contract_t$-tangle.

  Suppose first that $T^0$ is a star with centre $t_0$ and $|\beta^0(t)|=1$ for all $t\in
  V(T^0)\setminus\{t_0\}$. Then $\kappa\contract_{t_0}=\kappa$ (up to
  renaming of the elements of $U$ in $U\contract_{t_0}$). As there is at most
  one $\ell$-maximal $\kappa$-tangle $\CT_0$ with
  $\tilde\beta^0(t,t_0)\in\CT_0$ for every neighbour $t$ of $t_0$, there
  is at most one $\ell$-maximal  $\kappa\contract_{t_0}$-tangle. For
  all $t\neq t_0$, we have $|U\contract_t|=2$, and this implies that
  there is at most one $\ell$-maximal $\kappa\contract_t$-tangle.

  So suppose that $T^0$ is not such a star. Then $|U\contract_t|<|U|$
  for all $t\in V(T^0)$.
  For every node $t\in V(T^0)$ such that there is more than one  $\ell$-maximal
  $\kappa\contract_t$-tangle, we recursively apply the algorithm to
  $\kappa\contract_t$ and obtain a tree decomposition $(T_t,\beta_t)$
  of $\kappa\contract_t$. We let $\CN_t:=\CN(T_t,\beta_t)$. For nodes
  $t$ such that there is only one  $\ell$-maximal
  $\kappa\contract_t$-tangle, we let $\CN_t:=\emptyset$, and we let
  \[
  \CN:=\CN^0\cup\bigcup_{t\in V(T^0)}\CN_t.
  \]
  Using similar arguments as in the proof of Claim~1 in the proof of
  Theorem~\ref{theo:candec}, it is easy to prove that $\CN$ is nested
  and closed under complementation.

  We apply the algorithm of Lemma~\ref{lem:nested} to $\CN$ and obtain
  a tree decomposition $(T,\beta)$ of $\kappa$ with
  $\CN(T,\beta)=\CN$. It is easy to see that this decomposition has
  the desired properties.
\end{proof}

\section{Directed Decompositions}

In this section, we prove a variant of our canonical decomposition
theorem (Theorem~\ref{theo:candec}) in which we get rid of the hub
nodes at the price weakening the separation properties of the
decomposition and loosing a bit of the canonicity (the decomposition
will only be canonical given one tangle). This version of the
decomposition theorem is used in \cite{grosch15b} to design a
polynomial isomorphism test for graph classes of bounded rank width.

We work with a directed version of tree decompositions here. A
\emph{directed tree} is an oriented tree where all edges are directed
away from the root. In a directed tree $T$, by $\dagle^T$, or just
$\dagle$ if $T$ is clear from the context, we denote
the ``descendant order'', that is, the reflexive transitive closure
of the edge relation. The set of children of node $t\in V(T)$ is
denoted by $N_+^T(t)$ or just $N_+(t)$.

A \emph{directed tree decomposition} of a set $U$ or a connectivity
function $\kappa$ on $U$ is a pair $(T,\gamma)$, where $T$ is a
directed tree and $\gamma:V(T)\to 2^U$ such that 
\begin{itemize}
\item $\gamma(r)=U$ for the
root $r$ of $T$;
\item $\gamma(t)\supseteq\gamma(u)$ for all $(t,u)\in E(T)$;
\item $\gamma(u_1)\cap\gamma(u_2)=\emptyset$ for all siblings
  $u_1,u_2$. 

  (We call $u_1,u_2$ \emph{siblings} if $u_1\neq u_2$ and
  there is a $t\in V(T)$ such that $u_1,u_2\in N_+(t)$.)
\end{itemize}
We call $\gamma(t)$ the \emph{cone} of the decomposition at node
$t$. We define $\beta:V(T)\to 2^U$ by
\begin{equation}
  \label{eq:1}
  \beta(t):=\gamma(t)\setminus\bigcup_{u\in N_+(t)}\gamma(u).
\end{equation}
We call $\beta(t)$ the \emph{bag} of the decomposition at node
$t$. Observe that the bags are mutually disjoint and that their union
is $U$. Thus if $T_{\circ}$ is the undirected tree underlying $T$, then
$(T_{\circ},\beta)$ is an (undirected) tree decomposition in the sense
defined before. Moreover, for all $(s,t)\in E(T)$ we have
$\gamma(t)=\tilde\beta(s,t)$. 

Conversely, let $(T_{\circ},\beta)$ be an
undirected tree decomposition. Let $T$ a directed tree with underlying
undirected tree $T_{\circ}$ (obtained by arbitrarily choosing a root and
directing all edges away from the root) and define $\gamma:V(T)\to
2^U$ by
\begin{equation}
  \label{eq:2}
  \gamma(t):=\bigcup_{u\dagri t}\beta(u).
\end{equation}
Then $(T,\gamma)$ is a directed tree decomposition.

We always denote the bag function of a directed tree decomposition
$(T,\gamma)$ by $\beta$, and we use implicit naming conventions by
which, for example, we denote the bag function of $(T',\gamma')$ by $\beta'$. 

Now let $\KT$ be a family of mutually incomparable $\kappa$-tangles. A
\emph{directed tree decomposition for $\KT$} is a triple
$(T,\gamma,\tau)$, where $(T,\gamma)$ is a directed tree decomposition
of $\kappa$ and $\tau:\KT\to V(T)$ a bijective mapping such that the
following two conditions are satisfied.
  \begin{nlist}{DTD}
  \item\label{li:dtd1} For all nodes $t,u\in V(T)$ with $u\not\dagle t$ there is a minimum $(\tau^{-1}(u),\tau^{-1}(t))$-separation
    $Y$ such that $\gamma(u)\subseteq Y$.
\item\label{li:dtd2} For all nodes $t\in V(T)$ except the root, there is a node
    $u\in V(T)$ such that $t\not\dagle u$ and $\gamma(t)$ is a
    leftmost minimum $(\tau^{-1}(t),\tau^{-1}(u))$-separation.
  \end{nlist}
Observe that \ref{li:dtd1} implies that for all nodes $t\in V(T)$ and children $u\in N_+(t)$ we have
    $\gamma(u)\not\in\tau^{-1}(t)$. Furthermore, \ref{li:dtd2} implies
    that $\gamma(t)\in\tau^{-1}(t)$.
 
\begin{theorem}\label{theo:dcandec}
  Let $\ell\ge0$. Then there is a polynomial time algorithm that,
  given oracle access to a connectivity function $\kappa$ and a $\kappa$-tangle
  $\CT_{\textup{root}}\in \KT^{\le\ell}_{\max}$ (via a membership oracle or its index in a
  comprehensive tangle data structure for $\kappa$), computes a
  canonical directed tree decomposition $(T,\gamma,\tau)$ for the set
  $\KT^{\le\ell}_{\max}$ such
  that $\tau^{-1}(r)=\CT_{\textup{root}}$ for the root $r$ of $T$.
\end{theorem}

Here canonical means that if $\kappa':2^{U'}\to\NN$ is another
connectivity function and $\CT'_{\textup{root}}$ an $\ell$-maximal
$\kappa'$-tangle, and $(T',\gamma',\tau')$ is the decomposition
computed by our algorithm on input $(\kappa',\CT'_{\textup{root}})$,
then for every isomorphism $f$ from
$(\kappa,\CT_{\textup{root}})$ to $(\kappa',\CT'_{\textup{root}})$,
that is, bijective mapping $f:U\to U'$ with $\kappa(X)=\kappa'(f(X))$
and $X\in \CT_{\textup{root}}\iff f(X)\in \CT'_{\textup{root}}$ for
all $X\subseteq U$, there is an isomorphism $g$ from $T$ to $T'$ such
that that $f(\gamma(t))=\gamma'(g(t))$ for all $t\in V(T)$ and
$X\in\tau^{-1}(t)\iff f(X)\in(\tau')^{-1}(g(t))$ for all $X\subseteq
U,t\in V(T)$. 

\begin{proof}[Proof of Theorem~\ref{theo:dcandec}]
  Without loss of generality we assume that $|\KT^{\le {\ell}}_{\max}|\ge
  2$. Let $\CT_{\text{root}}\in\KT^{\le {\ell}}_{\max}$.

  We start our construction from a canonical undirected tree
  decomposition $(T_{\circ},\beta_{\circ},\tau_{\circ})$ for $\KT^{\le {\ell}}_{\max}$, which we compute by
  the algorithm of Theorem~\ref{theo:candec}. We let $r:=\tau_{\circ}(\CT_{\text{root}})$
  and henceforth think of the tree $T_{\circ}$ as being rooted in $r$. We
  denote the descendant order in this rooted tree by $\dagle_{\circ}$.
  For
  all $t\in V(T_{\circ})$, we let $\gamma_{\circ}(t):=\bigcup_{u\dagri_{\circ}
    t}\beta_{\circ}(u)$. Observe that $\gamma_{\circ}(r)=U$ and
  $\gamma_{\circ}(t)=\tilde\beta_{\circ}(s,t)$ for all nodes $t$ with parent $s$.

  We let
  \[
  V:=\tau_{\circ}(\KT^{\le {\ell}}_{\max})=\{\tau_{\circ}(\CT)\mid \CT\in\KT^{\le {\ell}}_{\max}\},
  \]
  and we let $T^{(0)}$ be the directed tree with vertex set $V(T^{(0)}):=V$ and
  edge set 
  \[
  E(T^{(0)}):=\big\{(t,u)\in V^2\bigmid t\dagsle_{\circ} u\text{ and there
      is no $x\in V$ such that }t\dagsle_{\circ} x\dagsle_{\circ} u\big\}.
  \]
  Let $\dagle^{(0)}$ be the descendant order in $T^{(0)}$. Obviously, $\dagle^{(0)}$
  is the restriction of $\dagle_{\circ}$ to $V$.

  We let $\tau:=\tau_{\circ}$. Note that $\tau$ is a bijection between
  $\KT^{\le {\ell}}_{\max}$ and $V$. For all $t\in V$ we let
  $\CT_t:=\tau^{-1}(t)$.

  Observe that for all $t\in V(T)$ we have
  $\gamma_{\circ}(t)\in\CT_t$: if $t=r$ is the root, we have
  $\gamma_{\circ}(t)=U\in\CT_t$, and if $s$ is the parent of $t$, we have
  $\gamma_{\circ}(t)=\tilde\beta_{\circ}(s,t)\in\CT_t$ by \ref{li:td3}.

  We define $\gamma:V\to 2^U$ as follows: we let $\gamma(r):=U$ and
  for every node $t\in V\setminus\{r\}$ we let $\gamma(t)$ be the
  leftmost minimum $(\CT_t,\bar{\gamma_{\circ}(t)})$-separation. As
  $\gamma_\circ(t)$ is a $(\CT_t,\bar{\gamma_{\circ}(t)})$-separation,
  there is a unique leftmost minimum
  $(\CT_t,\bar{\gamma_{\circ}(t)})$-separation. 
Then
  $\gamma(t)\in\CT_t$ and
  $\kappa(\gamma(t))\le\kappa(\gamma_{\circ}(t))$ and
  $\gamma(t)\subseteq\gamma_{\circ}(t)$. Note that we cannot just let
  $\gamma(t)=\gamma_\circ(t)$, because $\gamma_\circ(t)$ is not
  necessarily a \emph{leftmost} minimum
  $(\CT_t,\bar{\gamma_{\circ}(t)})$-separation, and \ref{li:dtd2}
  requires leftmost minimum separations.

  \begin{claim}
    For all $t\in V\setminus\{r\}$ there is a $t'\in V$ such that
    $t\not\dagle^{(0)} t'$ and $\gamma_{\circ}(t)$ is a minimum
    $(\CT_t,\CT_{t'})$-separation and $\gamma(t)$ is a leftmost minimum
    $(\CT_t,\CT_{t'})$-separation

    \proof Let $t\in V(T)\setminus\{r\}$ and $\CT:=\CT_t$. Let $s$ be the parent of
    $t$ in the rooted undirected tree $(T_{\circ},r)$. Then
    $\gamma_{\circ}(t)=\tilde\beta(s,t)$. By \ref{li:td2}, there are
    tangles $\CT'',\CT'$ such that the oriented edge $(t,s)$ appears
    on the oriented path from $t'':=\tau(\CT'')$ to $t':=\tau(\CT')$
    in $T_{\circ}$, and $\gamma_{\circ}(t)=\tilde\beta(s,t)$ is a minimum
    $(\CT'',\CT')$-separation. Note that $t',t''\in V(T)$ and $t\dagle
    t''$ and $t\not\dagle t'$. Furthermore, $\gamma_{\circ}(t)\in\CT$ is a
    $(\CT,\CT')$-separation.

    By \ref{li:td1}, there is an edge $(v',v)$ on the oriented path
    from $t'$ to $t$ in $T_{\circ}$ such that $\tilde\beta_{\circ}(v',v)$ is a
    minimum $(\CT,\CT')$-separation.
    Then
    \[
    \tilde\beta_{\circ}(v',v)\supseteq\tilde\beta(s,t)=\gamma_{\circ}(t)
    \]
    and $\kappa(\tilde\beta_{\circ}(v,v'))\le\kappa(\gamma_{\circ}(t))$.

    Suppose for contradiction that $\kappa(\tilde\beta_{\circ}(v',v))<
    \kappa(\gamma_{\circ}(t))$. As $\gamma_{\circ}(t)\subseteq
    \tilde\beta_{\circ}(v',v)$, this implies
    $\tilde\beta_{\circ}(v',v)\in\CT''$, and thus $\tilde\beta_{\circ}(v',v)$ is
    also a $(\CT'',\CT')$-separation. This contradicts $\gamma_{\circ}(t)$
    being a minimum $(\CT'',\CT')$-separation.

    Thus  $\kappa(\tilde\beta_{\circ}(v,v'))=\kappa(\gamma_{\circ}(t))$, and thus
    $\gamma_{\circ}(t)$ is a minimum $(\CT,\CT')$-separation as well.

    \medskip
    It remains to prove that $\gamma(t)$ is a leftmost minimum $(\CT,\CT')$-separation.
    Let $Y$ be the leftmost minimum
    $(\CT,\CT')$-separation. Then $\kappa(Y)=\kappa(\gamma_{\circ}(t))$ and $Y\subseteq\gamma_{\circ}(t)$.

    It follows that $Y$ is also the leftmost minimum
    $(\CT,\bar{\gamma_{\circ}(t)})$-separation, and thus $\gamma(t)=Y$. 
    \uend
  \end{claim}

  \begin{claim}[resume]
    For all $t,u\in V$ such that $t\dagsle^{(0)}u$ there is
    $Y\subseteq U$ such that $\gamma(u)\subseteq\gamma_{\circ}(u)\subseteq
    Y\subseteq\gamma_{\circ}(t)$ and $Y$ is a minimum
    $(\CT_u,\CT_t)$-separation.

    \proof
    By \ref{li:td1}, there is an edge
  $(t',u')$ on the path from $t$ to $u$ in $T_{\circ}$ such that
  $Y:=\gamma_{\circ}(u')=\tilde\beta(t',u')$ is a minimum
  $(\CT_u,\CT_t)$-separation. Then
  $\gamma(u)\subseteq \gamma_{\circ}(u)\subseteq \gamma_{\circ}(u')=Y\subseteq\gamma_{\circ}(t)$.\uend
  \end{claim}

  \begin{claim}[resume]
    For all $(t,u)\in E(T^{(0)})$, either $\gamma(u)\subseteq\gamma(t)$ or
    $\gamma(t)\cap\gamma(u)=\emptyset$. 

    \proof
    Let $(t,u)\in E(T^{(0)})$. If $t=r$ then we trivially have $\gamma(u)\subseteq
    U=\gamma(t)$. Therefore, we assume that $t\neq r$. Let $t'\in V$ such that
    $t\not\dagle^{(0)}t'$ and $\gamma_{\circ}(t)$ is a minimum
    $(\CT_t,\CT_{t'})$-separation and $\gamma(t)$ is a leftmost
    minimum $(\CT_t,\CT_{t'})$-separation. Such an $t'$ exists by
    Claim~1. Similarly, let $u'\in V$ such that
    $u\not\dagle^{(0)}u'$ and $\gamma_{\circ}(u)$ is a minimum
    $(\CT_u,\CT_{u'})$-separation and $\gamma(u)$ is a leftmost
    minimum $(\CT_u,\CT_{u'})$-separation.

    Suppose for contradiction that 
    $\gamma(u)\not\subseteq\gamma(t)$ and
    $\gamma(t)\cap\gamma(u)\neq\emptyset$, or equivalently,
    $\gamma(u)\not\subseteq\bar{\gamma(t)}$.
    \begin{cs}
      \case1
      $\kappa(\gamma(t))\le\kappa(\gamma(u))$.\\
      Then either $\gamma(t)\in\CT_u$ or $\bar{\gamma(t)}\in\CT_u$.
      \begin{cs}
        \case{1a}
        $\gamma(t)\in\CT_u$.\\
        If $\kappa(\gamma(t)\cap\gamma(u))\le\kappa(\gamma(u))$, then
      $\gamma(t)\cap\gamma(u)\in\CT_u$ and
      $\bar{\gamma(t)\cap\gamma(u)}\in\CT_{u'}$. Thus
      $\gamma(t)\cap\gamma(u)$ is a $(\CT_u,\CT_{u'})$-separation. As
      $\gamma(t)\cap\gamma(u)\subset\gamma(u)$ by our assumption
      $\gamma(u)\not\subseteq\gamma(t)$, this contradicts
      $\gamma(u)$ being a leftmost minimum
      $(\CT_u,\CT_{u'})$-separation. Thus $\kappa(\gamma(t)\cap\gamma(u))>\kappa(\gamma(u))$.

      By submodularity,
      \[
      \kappa(\gamma(t)\cup\gamma(u))<\kappa(\gamma(t))\le\kappa(\gamma_{\circ}(t)).
      \]
      Then
      $\gamma(t)\cup\gamma(u)\in\CT_t$. Furthermore,
      $\gamma(t)\cup\gamma(u)\subseteq
      \gamma_{\circ}(t)\cup\gamma_{\circ}(u)\subseteq\gamma_{\circ}(t)$. Thus
      $\bar{\gamma(t)\cup\gamma(u)}\supseteq\bar{\gamma_{\circ}(t)}\in\CT_{t'}$,
      and $\gamma(t)\cup\gamma(u)$ is a
      $(\CT_t,\CT_{t'})$-separation. This contradicts $\gamma_{\circ}(t)$ being a minimum
      $(\CT_t,\CT_{t'})$-separation. 
      \case{1b}
      $\bar{\gamma(t)}\in\CT_u$. \\
      If $\kappa(\bar{\gamma(t)}\cap{\gamma(u)})\le\kappa(\gamma(u))$, then
      $\bar{\gamma(t)}\cap\gamma(u)\in\CT_u$ and
      $\bar{\bar{\gamma(t)}\cap\gamma(u)}\in\CT_{u'}$. Thus
      $\bar{\gamma(t)}\cap\gamma(u)$ is a $(\CT_u,\CT_{u'})$-separation. As
      $\bar{\gamma(t)}\cap\gamma(u)\subset\gamma(u)$ by our assumption
      $\gamma(u)\not\subseteq\bar{\gamma(t)}$, this contradicts
      $\gamma(u)$ being a leftmost minimum
      $(\CT_u,\CT_{u'})$-separation. Thus $\kappa(\bar{\gamma(t)}\cap\gamma(u))>\kappa(\gamma(u))$

      By posimodularity, 
      \[
      \kappa(\gamma(t)\cap\bar{\gamma(u)})
      <\kappa(\gamma(t)).
      \]
      This implies
      $\bar{\gamma(t)\cap\bar{\gamma(u)}}=\bar{\gamma(t)}\cup\gamma(u)\in\CT_{t'}$.
      By Claim~2, there exists a minimum $(\CT_u,\CT_t)$-separation
      $Y\supseteq \gamma(u)$. We have
      \[
      \gamma(t)\cap\bar Y\cap(\bar{\gamma(t)}\cup\gamma(u))\subseteq
      \gamma(t)\cap\bar{\gamma(u)}\cap(\bar{\gamma(t)}\cup\gamma(u))=\emptyset.
      \]
      Thus $\gamma(t)\cap\bar{\gamma(u)}
      =\bar{\bar{\gamma(t)}\cup\gamma(u)}\in\CT_t$, and
      $\gamma(t)\cap\bar{\gamma(u)}$ is a
      $(\CT_t,\CT_{t'})$-separation. As $\kappa(\gamma(t)\cap\bar{\gamma(u)})
      <\kappa(\gamma(t))$, this contradicts $\gamma(t)$
      being a minimum $(\CT_t,\CT_{t'})$-separation.
      \end{cs}
      
      \case2
      $\kappa(\gamma(t))>\kappa(\gamma(u))$.\\
      Then either $\gamma(u)\in\CT_t$ or $\bar{\gamma(u)}\in\CT_t$. As
      above, let $Y\supseteq \gamma(u)$ be a minimum
      $(\CT_u,\CT_t)$-separation. 

      Then $\bar Y\in\CT_t$, and as $\bar Y\cap\gamma(u)=\emptyset$,
      we have $\bar{\gamma(u)}\in\CT_t$.

        If $\kappa(\gamma(t)\cap\bar{\gamma(u)})\le\kappa(\gamma(t))$,
        then $\gamma(t)\cap\bar{\gamma(u)}$ is a
        $(\CT_t,\CT_{t'})$-separation. As
        $\gamma(t)\cap\bar{\gamma(u)}\subset\gamma(t)$, this
        contradicts $\gamma(t)$ being a leftmost minimum $(\CT_t,\CT_{t'})$-separation.  Hence
        $\kappa(\gamma(t)\cap\bar{\gamma(u)})>\kappa(\gamma(t))$. 

        By posimodularity,
        \[
        \kappa(\bar{\gamma(t)}\cap\gamma(u))<\kappa(\gamma(u)).
        \]
        We
        have
        $
        \bar{\bar{\gamma(t)}\cap\gamma(u)}=\gamma(t)\cup\bar{\gamma(u)}\in\CT_{u'},
        $
        and as $\gamma(u)$ is a minimum $(\CT_u,\CT_{u'})$-separation,
        this implies
        $\bar{\gamma(t)}\cap\gamma(u)\not\in\CT_{u}$. Hence
        $\gamma(t)\cup\bar{\gamma(u)}\in \CT_u$. 

        Suppose for contradiction that
          $\kappa(\gamma(t)\cap\gamma(u))\le\kappa(\gamma(u))$.
          Then either $\gamma(t)\cap\gamma(u)\in\CT_u$ or
          $\bar{\gamma(t)\cap\gamma(u)}\in\CT_u$, and as
          \[
          (\gamma(t)\cup\bar{\gamma(u)})\cap\gamma(u)
          \cap\bar{\gamma(t)\cap\gamma(u)}
          =\gamma(t)\cap\gamma(u)
          \cap\bar{\gamma(t)\cap\gamma(u)}=\emptyset,
          \]
          we have $\gamma(t)\cap\gamma(u)\in\CT_u$. Now we continue
          exactly as in Case~1a: w e have 
          $\bar{\gamma(t)\cap\gamma(u)}=\bar{\gamma(t)}\cup\bar{\gamma(u)}\in\CT_{u'}$,
          which contradicts $\gamma(u)$ being a leftmost minimum
          $(\CT_u,\CT_{u'})$-separation. Thus
          $\kappa(\gamma(t)\cap\gamma(u))>\kappa(\gamma(u))$.

      By submodularity,
      $\kappa(\gamma(t)\cup\gamma(u))<\kappa(\gamma(t))\le\kappa(\gamma_{\circ}(t))$. Then
      $\gamma(t)\cup\gamma(u)\in\CT_t$. Furthermore,
      $\gamma(t)\cup\gamma(u)\subseteq
      \gamma_{\circ}(t)\cup\gamma_{\circ}(u)\subseteq\gamma_{\circ}(t)$. Thus
      $\bar{\gamma(t)\cup\gamma(u)}\supseteq\bar{\gamma_{\circ}(t)}\in\CT_{t'}$,
      and $\gamma(t)\cup\gamma(u)$ is a
      $(\CT_t,\CT_{t'})$-separation. This contradicts $\gamma_{\circ}(t)$ being a minimum
      $(\CT_t,\CT_{t'})$-separation. 
      \uend
    \end{cs}
  \end{claim}

  To obtain a treelike decomposition, which is supposed to satisfy
  $\gamma(u)\subseteq\gamma(t)$ whenever $u$ is a child of $t$, we
  need to restructure the tree, moving children $u$ of a node $t$ with
  $\gamma(u)\not\subseteq\gamma(t)$, and thus
  $\gamma(u)\cap\gamma(t)=\emptyset$ by Claim~3, upwards in the tree and
  attaching them as children to the first ancestor $s$ of $t$ such
  that $\gamma(u)\subseteq\gamma(s)$. To
  avoid inconsistencies, we do this by induction.

  We shall define a sequence of directed trees $T^{(i)}$, for $i$ ranging from
  $0$ to some $h$ to be determined, such
  that $V(T^{(i)})=V$ and the following conditions
  are satisfied for all $i$:
  \begin{eroman}
     \item if $i\ge 1$, then the descendant order $\dagle^{(i)}$ of
       $T^{(i)}$ is a subset of the descendant order $\dagle^{(i-1)}$ of
       $T^{(i-1)}$, that is, if $t\dagle^{(i)} u$ then
       $t\dagle^{(i-1)} u$;
     \item for all $(t,u)\in E(T^{(i)})$, either
       $\gamma(u)\subseteq\gamma(t)$ or $\gamma(u)\cap\gamma(t)=\emptyset$.
  \end{eroman}
  We have already defined $T^{(0)}$, and it clearly satisfies
  conditions (i) and (ii): condition (i) is void for $i=0$ and condition (ii)
  follows from Claim~3.

  So let us assume we have defined
  $T^{(i)}$ for some $i\ge0$. Let us call a node $u\in
  V\setminus\{r\}$ \emph{bad} in $T^{(i)}$ if
  $\gamma(u)\not\subseteq\gamma(t)$ for the parent $t$ of $u$. If
  there are no bad nodes in $T^{(i)}$, we let $h:=i$ and stop the
  construction. Otherwise, let $U^{(i)}$ be the set of all
  $\dagle^{(i)}$-maximal bad nodes. Note that the nodes $u\in U^{(i)}$ are mutually
  incomparable with respect to the order $\dagle^{(i)}$.
  
  To define $T^{(i+1)}$, for every $u\in U^{(i)}$, we do the
  following:
  \begin{itemize}
  \item we delete the edge from $u$ to its parent;
  \item we let $s_u$ be the the last (that is, $\dagle^{(i)}$-maximal) node 
    on the path from the root $r$ to $u$ such that
    $\gamma(u)\subseteq\gamma(s_u)$ (such a node exists because
    $\gamma(u)\subseteq U=\gamma(r)$);
  \item we add an edge from $s_u$ to $u$.
  \end{itemize}
  Then $T^{(i+1)}$ satisfies condition (i), because
  \[
  \dagle^{(i+1)}=\dagle^{(i)}\setminus\big\{(v,w)\bigmid\exists u\in
  U^{(i)}:\;s_u\dagsle^{(i)}v\dagsle^{(i)}u\text{ and }u\dagle^{(i)}
  w\big\}.
  \]
  To see that $T^{(i+1)}$ satisfies (ii), let $(t,u)\in
  E(T^{(i+1)})$. If $u\in U^{(i)}$, then we have $t=s_u$ and thus
  $\gamma(u)\subseteq\gamma(t)$. Otherwise
  $(t,u)\in E(T^{(i)})$, and we have $\gamma(u)\subseteq\gamma(t)$ or
  $\gamma(u)\cap\gamma(t)=\emptyset$ by the induction hypothesis.

  The construction terminates after at most height of $T^{(0)}$ many
  steps, because in each step the maximum height of a bad node decreases.

  We let $T:=T^{(h)}$. 

  \begin{claim}[resume]
    $(T,\gamma)$ is a directed tree decomposition of $\kappa$.

    \proof
    It is immediate from the construction that $T$ is a directed tree,
    that $\gamma(u)\subseteq\gamma(t)$ for all $(t,u)\in E(T)$, and
    that $\gamma(r)=U$ for the root $r$. It remains to prove that for
    all $t\in V$ and distinct children $u_1,u_2$ in $T$ we have
    $\gamma(u_1)\cap\gamma(u_2)=\emptyset$. Indeed, it follows from
    (i) that $t\dagle^{(0)} u_1,u_2$ and thus $t\dagle_{\circ} u_1,u_2$. Let
    $t'\in V(T_{\circ})$ be $\dagle_{\circ}$-maximal such that $t\dagle_{\circ} t'\dagle_{\circ}
    u_1,u_2$, and for $i=1,2$, let $u_i'$ be the child of $t'$ in
    $T_{\circ}$ such that $t'\dagsle_{\circ} u_i'\dagle_{\circ} u_i$. Then 
    \[
    \gamma(u_i)\subseteq\gamma_{\circ}(u_i)\subseteq\gamma_{\circ}(u_i')=\tilde\beta_{\circ}(t',u_i')
    \]
    and $\tilde\beta_{\circ}(t',u_1')\cap
    \tilde\beta_{\circ}(t',u_2')=\emptyset$. Thus
    $\gamma(u_1)\cap\gamma(u_2)=\emptyset$.
    \uend
  \end{claim}

  We have defined $\tau$ and $V=V(T)$ such that $\tau$ is a bijection between $\KT^{\le
    {\ell}}_{\max}$ and $V$. It remains to prove that $(T,\gamma,\tau)$
  satisfies \ref{li:dtd1} and \ref{li:dtd2}. Recall that
  $\tau^{-1}(t)=\CT_t$ and $\tau^{-1}(u)=\CT_u$.

  To prove \ref{li:dtd2}, let $t\in V\setminus\{r\}$. By Claim~1,
  there is a $u\in V$ such that $t\not\dagle^{(0)}u$ and $\gamma(t)$ is a
  leftmost minimum $(\CT_t,\CT_u)$-separation. As
  $\dagle$ is a subset of $\dagle^{(0)}$, we also have
  $t\not\dagle u$.

  To prove \ref{li:dtd1}, let $t,u\in V(T)$ such that $u\not\dagle t$.
  If $t\dagsle u$ then $t\dagsle^{(0)}u$, and thus by Claim~2, there is a
  $Y\supseteq\gamma(u)$ that is a minimum
  $(\CT_u,\CT_t)$-separation. In the following, we
  assume that $t\not\dagle u$. 
  Then
  $\gamma(t)\cap\gamma(u)=\emptyset$, because $(T,\gamma)$ is a
  directed tree decomposition. Moreover, $\gamma(t)\in\CT_t$
  and $\gamma(u)\in\CT_u$ by \ref{li:dtd2}. 

  Let $Y$ be a minimum
  $(\CT_u,\CT_t)$-separation. Suppose for contradiction that
  $\kappa(Y\cup\gamma(u))<\kappa(Y)$. Then
  $Y\cup\gamma(u)\in\CT_u$. Moreover, either $Y\cup\gamma(u)\in\CT_t$
  or $\bar{Y\cup\gamma(u)}\in\CT_t$. As
  $\gamma(u)\cap\gamma(t)=\emptyset$, we have
  \[
  \bar Y\cap \gamma(t)\cap (Y\cup\gamma(u))=\emptyset
  \]
  and thus $\bar{Y\cup\gamma(u)}\in\CT_t$. Hence $Y\cup\gamma(u)$ is a
  $(\CT_u,\CT_t)$-separation of order less than $\kappa(Y)$. This
  contradicts the minimality of $Y$.

  Hence  $\kappa(Y\cup\gamma(u))\ge\kappa(Y)$ and therefore, by
  submodularity, $\kappa(Y\cap\gamma(u))\le\kappa(\gamma(u))$. Thus
  $Y\cap\gamma(u)\in\CT_u$, and $Y\cap\gamma(u)$ is a
  $(\CT_u,\gamma_{\circ}(u))$-separation. By the minimality of $\gamma(u)$,
  it follows that $\gamma(u)\subseteq Y\cap\gamma(u)$ and thus
  $\gamma(u)\subseteq Y$.
\end{proof}

\section{Conclusions}
Our main contribution is to make tangles, an important tool in
structural graph theory, algorithmic.

The running time of all of our algorithms is $n^{O(k)}$,
where $k$ is the order of the tangles and $n$ the size of the ground
set. This running time is mainly caused by searching through all
potential bases for some separation. It is a very interesting open
question if we can do the same with fixed-parameter tractable
algorithms. This would even be interesting for specific connectivity
functions such as the connectivity function $\kappa_G$ and the
cut-rank function $\rho_G$ for graphs $G$ or the connectivity function
of representable matroids. Possibly, this can be achieved by arguments
building on Hlinen{\'{y}} and Oum's fpt algorithm for computing branch width
\cite{hlioum08}.

Even though the canonical decomposition theorem has useful
applications as it is, specifically in graph isomorphism testing, it
would be even more useful if one actually understood
it. However, except for a few simple cases of low order, the parts
$\beta(t)$ of the decompositions remain mysterious, for tangle nodes
$t$ and even more so for hub nodes $t$. For the connectivity function
$\kappa_G$ of graphs, Carmesin et al.~\cite{cardiehar+13b} obtained
first results clarifying the structure of the decompositions.


\end{document}